\documentclass[a4paper]{article}
\usepackage[margin=15truemm]{geometry}
\usepackage{graphicx}
\usepackage{xcolor}
\usepackage{multicol}
\usepackage[colorlinks=true,linkcolor=blue,citecolor=blue]{hyperref}
\makeatletter
\def\hyper@linkstart#1#2{\begingroup\color{\csname @#1color\endcsname}}
\def\hyper@linkend{\endgroup}
\def\hyper@link#1#2#3{{\color{\csname @#1color\endcsname}#3}}
\makeatother

\usepackage[super,numbers,sort&compress]{natbib}
\usepackage{amssymb}
\usepackage{amsthm}
\usepackage[subrefformat=parens]{subcaption}
\usepackage{enumitem}
\usepackage{url}
\usepackage{abstract}
\usepackage{newtxtext,newtxmath}
\usepackage{float}
\usepackage{titlesec}

\newenvironment{sfigure}{%
  \begin{figure}[H]%
}{%
  \end{figure}%
}

\newenvironment{lfigure}{%
  \begin{figure*}%
}{%
  \end{figure*}%
}

\usepackage{amsmath}

\newcommand{\defrefnum}[2]{\expandafter\def\csname refnum@#1\endcsname{#2}}

\newcommand{\myref}[1]{%
  \ifcsname refnum@#1\endcsname
    \csname refnum@#1\endcsname
  \else
    \ref{#1}%
  \fi
}

\newcommand{\silentref}[1]{\iftrue\ref{#1}\fi}

\newcommand{\myciten}[1]{\citenum{#1}}

\newcommand{\addedc}[1]{{\color{black} {#1}}}

\newcommand{\AL}[1]{\begin{align}#1\end{align}}
\newcommand{\LN}[1]{\label{#1}\\}
\newcommand{\LB}[1]{\label{#1}}
\newcommand{\NN}{\nonumber\\}
\newcommand{\TW}[1]{}
\newcommand{\AS}[1]{\left(#1\right)}
\newcommand{\AD}[1]{\left[#1\right]}
\newcommand{\vect}[1]{{\boldsymbol{#1}}}

\newcommand{\tr}{^{\top}}

\newcommand{\Real}{\mathbb{R}}
\newcommand{\Inte}{\mathbb{Z}}
\newcommand{\Nnum}{\mathbb{N}}
\newcommand{\iunit}{\mathrm{i}}

\newcommand{\vdi}{\vect{d}_i}
\newcommand{\vdj}{\vect{d}_j}

\newcommand{\sumi}{\sum_{i\in\viset}}
\newcommand{\sumj}{\sum_{j\in\viset}}

\newcommand{\sumij}{\sum_{i, j\in\viset}}

\newcommand{\argmax}{\mathop{\textrm{arg\;max}}}
\newcommand{\kdim}{\prod_{\idim=1}^{K}}
\newcommand{\ddim}{\prod_{\idim=1}^{D}}

\newcommand{\dprod}{\prod_{\idim=1}^{D}}

\newcommand{\ksum}{\sum_{\idim=1}^{K}}

\newcommand{\lsh}{L}
\newcommand{\lsp}{L}
\newcommand{\lpa}{L'}
\newcommand{\ran}{R}
\newcommand{\vlsh}{\vect{\lsh}}
\newcommand{\vlpa}{\vect{\lpa}}
\newcommand{\vran}{\vect{\ran}}
\newcommand{\grp}[1]{_{[#1]}}
\newcommand{\idim}{k}
\newcommand{\lshd}{\lsh_\idim}
\newcommand{\lpad}{\lpa_\idim}
\newcommand{\rand}{\ran_\idim}
\newcommand{\vrel}{\vect{\rel}}
\newcommand{\vpos}{\vect{\pos}}
\newcommand{\rel}{r}
\newcommand{\pos}{d}
\newcommand{\ransp}{\mathcal{R}}

\newcommand{\lshsp}{\mathcal{L}}

\newcommand{\dimth}{^{(D)}}

\newcommand{\mkth}{^{(K-1)}}
\newcommand{\lkth}{^{(K)}}

\newcommand{\cfil}{\mathcal{C}}
\newcommand{\ravl}{\tilde{g}}
\newcommand{\ravlr}{\tilde{h}}
\newcommand{\ravm}{g}
\newcommand{\ravr}{h}
\newcommand{\ravx}{h'}
\newcommand{\mymod}[2]{{{#1}\bmod{#2}}}
\newcommand{\sgn}{\mathrm{sgn}}

\newcommand{\minimize}[1]{{\underset{#1}{\textrm{minimize}}}}
\newcommand{\maximize}[1]{{\underset{#1}{\textrm{maximize}}}}
\newcommand{\st}{{\textrm{subject to}}}

\newcommand{\Inteco}[1]{{[#1)}}

\newcommand{\InteN}[1]{{\Inteco{#1}}}

\newcommand{\spos}{\vect{k}}

\newcommand{\amp}{c}

\newcommand{\inten}{I}
\newcommand{\refinten}{I_R}
\newcommand{\frefinten}{\tilde{I}_R}
\newcommand{\flen}{\mathrm{f}}
\newcommand{\ssize}{W}
\newcommand{\plcost}{c}

\newcommand{\imset}{I}
\newcommand{\viset}{\mathcal{N}}
\newcommand{\indi}{\mathbb{I}}
\newcommand{\fline}{f^{(\text{line})}}
\newcommand{\fblob}{f^{(\text{blob})}}
\newcommand{\cmax}{c^{(\text{max})}}
\newcommand{\coline}{c^{(\text{line})}}
\newcommand{\coblob}{c^{(\text{blob})}}
\newcommand{\sline}{\sigma^{(\text{line})}}
\newcommand{\sblob}{\sigma^{(\text{blob})}}
\newcommand{\nline}{N^{(\text{line})}}
\newcommand{\nblob}{N^{(\text{blob})}}

\newcommand{\fou}[1]{\mathcal{F}[{#1}]}

\newcommand{\ifou}[1]{\mathcal{F}^{-1}[{#1}]}
\newcommand{\ifouAS}[1]{\mathcal{F}^{-1}\AD{#1}}

\newcommand{\hamil}{H}

\newcommand{\obj}{F}
\newcommand{\gidx}{g}
\newcommand{\vx}{\vect{x}}
\newcommand{\vy}{\vect{y}}
\newcommand{\vz}{\vect{z}}

\newcommand{\mycite}{\cite}

\theoremstyle{thmstyleone}%
\newtheorem{theorem}{Theorem}%

\theoremstyle{thmstyletwo}%
\newtheorem{prop}[theorem]{Proposition}%

\theoremstyle{thmstylethree}%
\newtheorem{dfn}{Definition}%

\renewcommand\thesection{}
\renewcommand\thesubsection{}

\makeatletter
\renewcommand{\@seccntformat}[1]{}
\renewcommand\section{\@startsection {section}{1}{\z@}%
                                   {-3.5ex \@plus -1ex \@minus -.2ex}%
                                   {2.3ex \@plus.2ex}%
                                   {\normalfont\large\bfseries}}
\renewcommand\subsection{\@startsection{subsection}{2}{\z@}%
                                     {-1.25ex\@plus -1ex \@minus -.2ex}%
                                     {-2.3ex \@plus .2ex}%
                                     {\normalfont\normalsize\bfseries}}
\renewcommand\subsubsection{\@startsection{subsubsection}{3}{\z@}%
                                     {-1.25ex\@plus -1ex \@minus -.2ex}%
                                     {-2.3ex \@plus .2ex}%
                                     {\normalfont\normalsize\bfseries}}
\makeatother

\begin{document}


\title{Convolutional Formulation of Large-Scale Quadratic Unconstrained Binary Optimization with Dense Interactions}
\author{Hiroshi Yamashita and Hideyuki Suzuki}
\maketitle

\fontsize{10pt}{11pt}\selectfont

\begin{abstract}
\addedc{
The spatial photonic Ising machine (SPIM) is a promising optical hardware solver for large-scale combinatorial optimization problems with dense interactions. As the SPIM can represent Ising problems with rank-one coupling matrices, multiplexed versions have been proposed to enhance applicability to higher-rank interactions. However, the multiplexing cost reduces implementation efficiency, and even without multiplexing, the SPIM can represent coupling matrices beyond rank-one. To clarify the intrinsic representation power of the SPIM, we propose spatial quadratic unconstrained binary optimization (spQUBO), a formulation of Ising problems with spatially convolutional structures. We prove that any spQUBO reduces to a two-dimensional spQUBO with the convolutional structure preserved, which can be efficiently implemented on the SPIM without multiplexing. We demonstrate its applicability to distance-based combinatorial optimization, including placement problems and clustering problems. These results advance our understanding of the class of optimization problems where SPIMs exhibit unique advantage in efficiency and scalability. Furthermore, the convolutional structure of spQUBO also enables efficient computation using Fast Fourier Transforms. 
}

\end{abstract}





\begin{multicols}{2}

\section{\addedc{Introduction}}

To meet the growing demand for high-performance and high-efficiency computing in data science and artificial intelligence, there have been active efforts to develop domain-specific computing systems. 
Ising solvers, often also called Ising machines, are dedicated hardware designed to solve Ising problems, or equivalently, quadratic unconstrained binary optimization (QUBO) problems, which have broad applicability to important combinatorial optimization problems \mycite{LucasNP}. 
The simplicity of this approach allows Ising solvers to be implemented using various physical phenomena \mycite{IMreviewMohseni}, 
including quantum effects \mycite{Dwave}, laser beams \mycite{CIMPRA2013,CIMScience2016,CIMAPL2020}, dynamical systems \mycite{MA,OIM,CAC,CACScaling,CBM,ChaoticIM25}, and advanced digital electronics technologies \mycite{DA,Statica,CBMAnnealer}. 

However, the physical implementation of large-scale Ising solvers that can handle dense interactions is not straightforward. 
As the number of design variables, referred to as spins, increases, the number of pairwise spin interactions grows quadratically, often making hardware implementation infeasible. 
Therefore, in some implementations the interaction structure is restricted to a specific class of sparse networks \mycite{Chimera,MassivelyParallel}, which requires that QUBOs with dense interactions be transformed accordingly, incurring additional overhead. 
To implement dense interactions, high-speed physical Ising solvers often rely on a large number of digital computing devices, including FPGAs and GPUs, as in the 100,000-spin implementations \mycite{SBM,100KCIM}. 
Thus, the scalability of many Ising solvers is constrained by the computational resources required to handle dense interactions. 

The spatial photonic Ising machine (SPIM) \mycite{PierangeliPRL,Gauge21,SPIMQuadrature,SPIMWavelength,SPIMSakabe,SPBM,SPIMYe2023v2,OguraCREST24,FocalDivision25,Circulant25,PhotonicSG25,Arbitrary25,Correlation21,Yao360K25} is an optical Ising solver that has superior scalability in handling large-scale Ising problems with dense interactions. 
By utilizing the spatial parallelism of light propagation, it is expected to exhibit superior efficiency to other Ising machines \mycite{FocalDivision25,Circulant25}. 
Despite these expectations, the primitive version of SPIM can represent Ising problems with only rank-one coupling matrices, which limits its applicability to real-world problems. 
Accordingly, multiplexed versions have been proposed to enhance the applicability to higher-rank interactions \mycite{SPIMQuadrature,SPIMWavelength,SPIMSakabe,SPBM,SPIMYe2023v2,OguraCREST24,FocalDivision25,Circulant25}, although multiplexing reduces the implementation efficiency. 

In this paper, we investigate the intrinsic representation power of the original SPIM without multiplexing. 
Although it is known to represent coupling matrices beyond rank-one \mycite{PierangeliPRL,Circulant25,PhotonicSG25}, we clarify its potential capabilities. 
Specifically, we introduce a new class of QUBO problems with spatially convolutional structures, termed spatial QUBO (spQUBO). 
We show that a specific subclass of spQUBO, the two-dimensional periodic spQUBO, can be efficiently implemented on the SPIM without multiplexing. 
Furthermore, we present a reduction algorithm that can transform any spQUBO into a two-dimensional periodic spQUBO while preserving the convolutional structure. 

Due to its optical nature, the SPIM architecture has an expectation of efficiency in representing spatially convolutional interactions. 
The spQUBO formulation reveals that the problem size that can be implemented on an SPIM system is determined by the spatial volume of the configuration domain, where the interactions between variables are represented; in other words, unlike ordinary Ising machines, the problem size is not directly limited by the number of variables. 
Therefore, to implement large-scale combinatorial optimization problems on SPIM without reducing its scalability and efficiency, it is crucial to focus on the convolutional structures and obtain compact spQUBO representations with minimal spatial volume. 
Many real-world combinatorial optimization problems are expected to have convolutional structures, e.g., those defined on spatially distributed variables, potentially leading to broad applications. 

These results advance our understanding of the class of optimization problems where SPIMs exhibit unique advantage in efficiency and scalability. 
Furthermore, spQUBO's efficiency is not limited to the SPIM architecture; we show that its convolutional structure allows efficient computation using Fast Fourier Transforms (FFT). 

\section{Results}

\subsection{Quadratic unconstrained binary optimization}

The Ising problem is an optimization problem for $N$ spin variables, as shown below: 
\AL{
\minimize{\sigma_1,\ldots,\sigma_N} \ \hamil &= -\frac{1}{2}\sumij  J_{ij} \sigma_i\sigma_j - \sumi h_i\sigma_i\LN{eq:Ising}
\st \ \sigma_i &\in \{-1,+1\}\quad (i \in \viset),
}
where $\viset=\{1,\ldots,N\}$ is the index set of the variables, $J$ is the coupling matrix, and $h$ is the bias vector.  
In this problem, the interactions between the spins are represented as the coupling matrix $J$. 

There are also situations where it is more natural to use $x_i\in\{0,1\}$ instead of $\sigma_i\in\{-1,+1\}$ as decision variables, or to consider the maximization problem instead of the minimization problem. 
In this case, we can consider 
\AL{
\maximize{x_1,\ldots,x_N}\ \obj &=\frac{1}{2} \sumij W_{ij} x_ix_j + \sumi b_i x_i \LN{eq:QUBO}
\st\ x_i &\in \{0, 1\}\  (i \in \viset),
}
where $W$ and $b$ are the coupling matrix and the bias vector, respectively. 
This problem is often referred to as a quadratic unconstrained binary optimization (QUBO). 
The Ising and QUBO formulations are equivalent under a simple change of variables, $x_i = (\sigma_i+1)/2$, and a sign flip, with a change in the parameters. 
Therefore, we use the two terms interchangeably, as is common in the literature. 

The objective function in the Ising problem corresponds to the Hamiltonian or the energy in the context of spin systems in statistical physics. 
Additionally, we can consider a probability model where the probability of the spin configuration is proportional to $\exp(-\hamil/T)$ with a temperature $T$. 
In such a model, if $J_{ij}$ or $W_{ij}$ is positive, the spin values $x_i$ and $x_j$ in the low-energy state tend to align. 
This coupling between spins is called ferromagnetic. 
In contrast, when $J_{ij}$ or $W_{ij}$ is negative, the coupling is called antiferromagnetic. 

In this study, we focus on the computation of quadratic terms of the Ising and QUBO Hamiltonians, because linear bias terms do not become a bottleneck even when computed on a digital computer. 

\subsection{SPIM}
\label{sec-back-spim}

The SPIM is an optical Ising solver that utilizes a spatial light modulator (SLM) for solving large-scale Ising problems \mycite{PierangeliPRL}. 
Fig.~\myref{fig:SLM} presents a schematic for the SPIM. 
The elements of the SLM are arranged in a two-dimensional grid, where the position of the $i$-th element is represented by an integer point $\vpos_i\in\Inte^2$. 
The SLM modulates the phase of the incoming laser light according to the spin value $\sigma_i$, whose amplitude $\xi_i$ is pre-modulated by another SLM. 
After propagating the modulated light through the optical system with a lens as shown in Fig.~\myref{fig:SLM}, we can observe the light intensity $\inten(\vect{x})$ with an image sensor. 
By comparing it to a reference image $\refinten(\vx)$, we have the value approximately equivalent to $\hamil=-(1/2)\sumij J_{ij}\sigma_i\sigma_j$, where $J$ is represented as 
\AL{
J_{ij}&=\xi_i\xi_j \frefinten(2W(\vpos_i-\vpos_j)),\LB{eq:SPIM-hamiltonian} 
}
where $W$ is a constant coefficient and $\frefinten$ is the Fourier transform of $\refinten$. 
It has the same form as the Ising Hamiltonian (\myref{eq:Ising}) and the coupling coefficients $J_{ij}$ are determined by the reference image $\refinten$. 
The primitive version of SPIM uses only the observation on the origin $\inten(\vect{0})$ and calculates Hamiltonian with rank-one coupling matrix $J_{ij}=\xi_i\xi_j$. 
We can multiplex this model to represent higher-rank interactions \mycite{SPIMQuadrature,SPIMWavelength,SPIMSakabe,SPBM,SPIMYe2023v2,OguraCREST24,FocalDivision25,Circulant25}. 
Further details can be found in the literature \mycite{PierangeliPRL,Correlation21} or in Methods. 

\begin{lfigure}
\captionsetup[sub]{labelformat=empty}
\begin{center}

\includegraphics[ width=5in ]{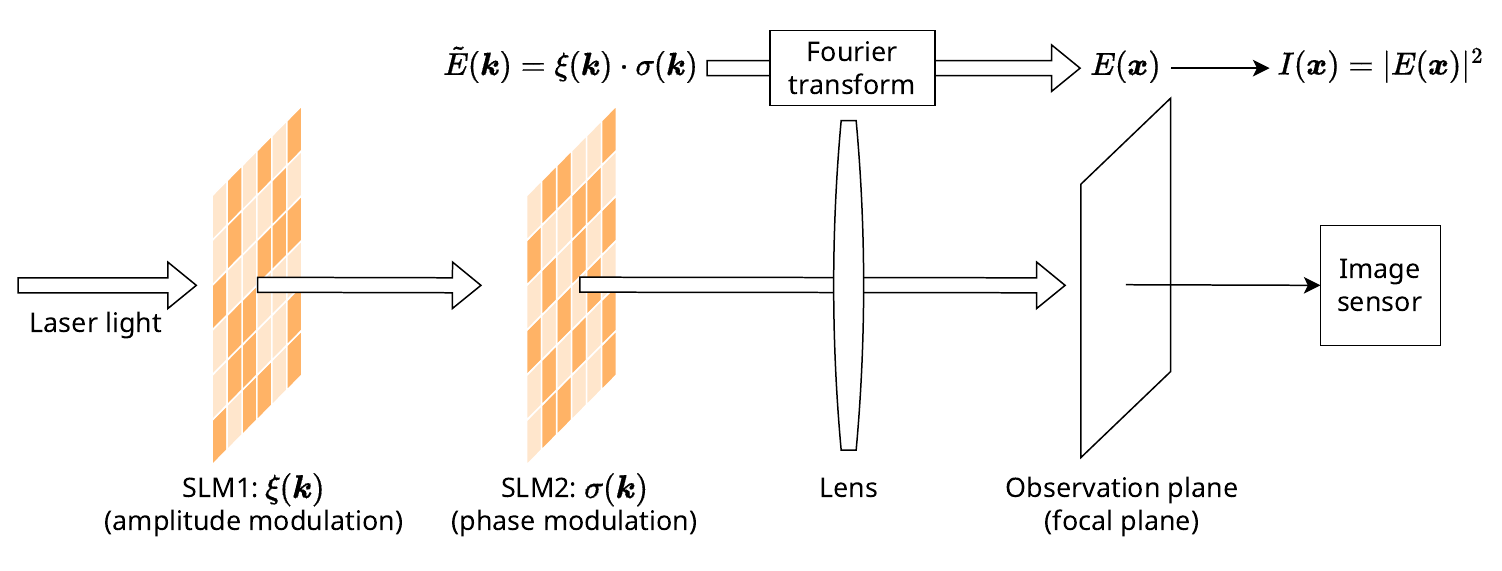}
\end{center}
\caption{
A schematic of the architecture of SPIM. 
\addedc{
The lower part indicates the components of the SPIM system, where the white and black arrows indicate the flow of laser light and information in the system, respectively. 
The upper part indicates the flow of computation, accompanied with the notation for the spatial pattern represented in each stage. 
}
}
\label{fig:SLM}
\end{lfigure}

\subsection{Spatial QUBO}
\label{sec-spqubo}

To formulate spatially convolutional interactions, we associate variables in a QUBO with grid points in a $D$-dimensional integer grid $\Inte^D$. 
We assume that these points are located in a finite region denoted as 
\AL{
\InteN{\vlsh} \equiv \{\vpos\in\Inte^D \mid \vect{0} \le \vpos < \vlsh \},
}
where $\vlsh\in\Nnum^D$ and the inequality signs for vectors require inequality of all the elements. 
Throughout this paper, we will denote vectors by bold symbols and scalars by normal symbols. An element of a vector is denoted by its symbol in a normal style with an index (e.g. $\lsh_d$ means the $d$-th element of the vector $\vlsh$). 
We formulate combinatorial optimization problems with a spatially convolutional structure as spatial QUBO (spQUBO): 

\begin{dfn}[spatial QUBO]
\label{dfn:spQUBO}
A ($D$-dimensional) spatial QUBO (spQUBO) is an optimization problem represented as 
\AL{
\maximize{x_1,\ldots,x_N}\ \obj&=\frac{1}{2}\sumij c_ic_jf(\vpos_i-\vpos_j)x_ix_j + \sum_ib_i x_i\LN{eq:spCM}
\st\ x_i &\in \{0, 1\}\quad (i \in \viset),
}
where $f\colon\Inte^D\to\Real$ is a symmetric function satisfying $f(\vpos)=f(-\vpos)$, each $i$-th variable $x_i$ is associated with a bias $b_i\in\Real$, a coefficient $\amp_i\in\Real$, and a grid point $\vpos_i\in\InteN{\vlsh}$ for $\vlsh\in\Nnum^D$, and the mapping $i\mapsto \vpos_i$ is injective. 
For an spQUBO, $f$ is called the spatial coupling function, $\vlsh$ is called its spatial shape, and the product $|\InteN{\vlsh}|=\prod_{d=1}^D\lsh_d$ is called its spatial volume. 
The region $\InteN{\vlsh}$ is called its configuration domain. 

\end{dfn}
Note that the symmetry condition $f(\vpos)=f(-\vpos)$ does not restrict the representation power, because from any non-symmetric spatial coupling function, we can always find a symmetric $f$ without changing the objective function $\obj$. 

We will also consider the case that spatial coupling functions are periodic. 
Specifically, a function $f$ on $\Inte^D$ is periodic with period $\vlsh\in\Nnum^D$ if $f(\vpos) = f(\vpos+\spos\odot\vlsh)$ holds for all $\vpos, \spos\in\Inte^D$, where $\odot$ denotes the Hadamard element-wise product. 
We denote two integer vectors satisfying $\vpos'=\vpos+\spos\odot\vlsh$ as $\vpos'=\vpos \pmod{\vlsh}$. 
A periodic function with period $\vlsh$ is fully determined by the values on the configuration domain $\InteN{\vlsh}$.  
For a grid point $\vpos\in \Inte^D$, we can uniquely find a corresponding point $\vpos'=\vpos \pmod{ \vlsh}$ in the configuration domain $\InteN{\vlsh}$. 

\begin{dfn}[periodic spatial QUBO]
\label{dfn:periodic-spQUBO}
A ($D$-dimensional) spatial QUBO is periodic if its spatial coupling function $f$ is periodic with period $\vlsh$ identical to its spatial shape.  

\end{dfn}
An spQUBO is a QUBO with the coupling matrix $W_{ij}=c_ic_jf(\vpos_i-\vpos_j)$. 
Conversely, any QUBO with coupling matrix $W_{ij}$ can be represented as an $N$-dimensional periodic spQUBO, because if we set $\vpos_i\in\{0, 1\}^{N}$ to be a one-hot vector, we can construct a function $f$ with period $\vlsh=(2, \ldots, 2)$. 
Specifically, the $i$-th grid point is defined as  
\AL{
(\vpos_i)_k &= \begin{cases}
1 & (k=i)\\
0 & (k\neq i)
\end{cases}
}
Then, in the transformed spQUBO, we set $\amp_i=\amp_j=1$ and  
\AL{
f(\vrel) &= \begin{cases} 
\frac{1}{2}(W_{ij}+W_{ji}) & (\vrel=\vpos_i-\vpos_j \pmod{\vlsh})\\
0 & (\text{otherwise})
\end{cases}.
}
Because of these mutual transformations, the problem classes QUBO, spQUBO, and periodic spQUBO are all the same. 

This equivalence seemingly implies that any QUBO can be implemented on an SPIM. 
However, in the construction as an spQUBO, its spatial volume $2^N$ grows exponentially with the number of variables $N$. 
Thus, it is crucial to find an efficient representation of a combinatorial optimization problem as an spQUBO for implementing it efficiently on an SPIM. 
To that end, in this paper, we make full use of the spatial convolutional structure inherent in problems to represent them as spQUBOs with less spatial volume. 

While we limit the description here to spatial QUBO, we can also define the spatial Ising problem accordingly. 

\subsection{Distance-based combinatorial optimizations}
\label{sec-distance}

The convolutional structure of spQUBO can efficiently represent interaction weights based on the relative positions of spins, especially the distances between them. 
Here, we present realistic problems with distance-based interactions as motivating examples for introducing spQUBO and demonstrating its applicability. 

Let us consider a placement problem to determine the optimal placement of facilities. 
We assume that the facilities can be placed on grid points on a two-dimensional region and the objective is to maximize the total utility minus the total cost of the placements. 
Here, placements are represented by binary variables $x_i \in \{0, 1\}$, where $x_i = 1$ if a facility is placed at the $i$-th grid point and $x_i = 0$ otherwise. 
It is natural to assume that each utility of the placement diminishes with the increase of neighboring facilities. 
If we represent such marginal utility as a quadratic function, the maximization of the total utility becomes an spQUBO. 
The interaction weights $W_{ij}$ are always negative, with their absolute values decreasing as the distance $\rel_{ij}$ increases. 
These antiferromagnetic couplings, stronger for nearby grid points, prevent redundant facilities from being placed too close to one another. 
The specific forms of the QUBO and its derivation are given in Methods. 

Next, we consider a clustering problem to partition a set of data points into clusters. 
The simple-cost method \mycite{SimpleCost,DBClustering} formulates clustering as a QUBO problem to minimize the sum of the distances between data point pairs within the same cluster. 
If we consider two clusters, the partitioning is represented by binary variables $x_i\in\{0, 1\}$, 
which indicate the cluster index to which the $i$-th data point is assigned. 
The total distance to be minimized is given by 
\AL{
2\Bigg(\sumij \rel_{ij}x_ix_j\Bigg) - 2\Bigg(\sumi x_i\sumj \rel_{ij}\Bigg) + \Bigg(\sumij \rel_{ij}\Bigg),
}
which defines a QUBO with coupling matrix and bias vector given by $W_{ij} = -2\rel_{ij}$ and $b_i = \sumj \rel_{ij}$. 
In this formulation, the interaction weights $W_{ij}$ are always negative, 
with their absolute values increasing as the distance $\rel_{ij}$ increases. 
These antiferromagnetic couplings tend to assign distant data points to different clusters, 
which is a desirable property for clustering. 
If the data points $\vpos_i$ and distances $\rel_{ij}$ are defined on grid points in $D$-dimensional space, the problem becomes a $D$-dimensional spQUBO. 
This can be easily extended to the general case of more than two clusters. 
The partitioning is represented by binary variables $x_{i\gidx}\in\{0, 1\}$, which indicate whether the $i$-th data point is assigned to the $g$-th cluster, and are associated with the point $(\vpos_i , \gidx)$ in $(D+1)$-dimensional space. 
The specific forms of the QUBO and its derivation are given in Methods. 

Restoration of binary images can also be formulated as a QUBO \mycite{Geman}. 
The images are represented by binary variables associated with pixels located at grid points in two-dimensional space. 
To recover a clean image from a noisy one, where each pixel is flipped with a certain probability, the interaction weights are set to be ferromagnetic, encouraging neighboring pixels to adopt the same value. 
Naturally, these interaction weights are determined based on the relative positions of the pixels, indicating that the QUBO is, in fact, an spQUBO. 

In QUBO formulations of combinatorial optimization problems \mycite{LucasNP}, exact-one constraints often appear as penalty terms to ensure that exactly one binary variable in a group is set to 1. 
For example, in the clustering problem described above, the exact-one constraint ensures that each data point is assigned to exactly one cluster. 
In general, for a set of variables $x_i$ with indices $i\in C$, the exact-one constraint requires that exactly one variable is set to 1; i.e., $\sum_{i\in C}x_i=1$. 
This constraint can be implemented in a QUBO by introducing a quadratic penalty term 
$\hamil=(\sum_{i\in C}x_i-1)^2$, whose minimization enforces the constraint.
Expanding the penalty term yields $H = \sum_{i,j\in C}x_ix_j - 2\sum_{i\in C}x_i + 1$, which can be represented as an spQUBO with a constant spatial coupling function. 
Therefore, the exact-one constraint can be naturally incorporated into spQUBO formulations, as seen in the clustering problem above. 

\subsection{SPIM and two-dimensional spQUBO}
\label{sec-spim}

In the following, we show the equivalence between two-dimensional periodic spQUBOs and QUBOs represented by the SPIM. 
It should be noted that we derive the equivalence for the idealized SPIM without accounting for physical constraints such as device precision, measurement accuracy, and control of other system parameters. 

On the SLM plane of the idealized SPIM, pixels are assumed to be located on square grid points with a pitch of $\ell$. 
The physical location of pixel $\spos_j$ associated with the $j$-th spin is expressed as $\spos_j = \ell\vpos_j$, where $\vpos_j$ is an integer vector. 
On the observation plane as shown in Fig.~\myref{fig:SLM}, an image sensor measures intensity at pixels arranged on square grid points with a pitch of $a$. 
As proposed in Ref.~\myciten{Correlation21}, the inner product of the measured intensity and the reference function $\refinten$, which can take negative values, yields the Hamiltonian value. In the above setting, it can be represented as a comb-shaped function 
\AL{
\refinten(\vx) = \sum_{\vect{n}\in\Inte^2}\refinten^{(\vect{n})}\delta(\vx-a\vect{n}),
}
where $\refinten^{(\vect{n})}=\refinten(a\vect{n})$ is a reference array that needs to be designed depending on the applications.  
The effective size of the SLM is determined by the optical setup of the SPIM, including the focal length $\flen$ of the lens and the wavelength $\lambda$ of the laser; the spin locations are bounded by $\vpos_j\in\InteN{\lsp}^2$, where $\lsp=(\flen\lambda)/(a\ell)$ is designed to be an integer. 
The reference array can be assumed symmetric, $\refinten^{(\vect{n})}=\refinten^{(-\vect{n})}$, as only the symmetric part contributes to the Hamiltonian. 
Then, the Hamiltonian of the SPIM is expressed as 
\AL{
H \propto \sumij\sigma_i\sigma_j\xi_i\xi_j\int d\vx\refinten(\vx)\cos\Bigg(-\frac{2\pi\ell}{\flen\lambda}(\vpos_i-\vpos_j)^\top\vx\Bigg),\LB{eq:SPIM-hamiltonian-sym} 
}
as derived in Methods. 
With the comb-shaped reference function, we have 
\AL{
H \propto \sumij\sigma_i\sigma_j\xi_i\xi_j\sum_{\vect{n}\in\Inte^2}\refinten^{(\vect{n})}\cos\Bigg(-\frac{2\pi }{\lsp}(\vpos_i-\vpos_j)^\top\vect{n}\Bigg). 
}
This is a two-dimensional periodic spQUBO with the symmetric coupling function 
\AL{
f(\vrel)=\sum_{\vect{n}\in\Inte^2}\refinten^{(\vect{n})}\cos\Bigg(-\frac{2\pi }{\lsp}\vrel^\top\vect{n}\Bigg), 
}
which is an (inverse) discrete Fourier transform (DFT) of the symmetric reference array $\refinten^{(\vect{n})}$. 

Conversely, let us assume a two-dimensional periodic spQUBO with spatial shape $(\lsp, \lsp)$ is given. 
Then, we can find the reference array $\refinten^{(\vect{n})}$ by performing the (inverse) DFT of the coupling function $f(\vrel)$. 
Of course, we need to set the SPIM system parameters with $\lsp$ large enough to accommodate the spin locations of the problem. 
Thus, when formulating a problem, minimizing the spatial volume is important for the efficient implementation on the SPIM. 
We assumed the case of a square spatial shape $(\lsp, \lsp)$ here for simplicity; however, the non-square case where $\lsp=(\lsp_1, \lsp_2), \lsp_1\neq \lsp_2$ can also be handled by using different system parameters for each axis. 

\subsection{Reduction to two-dimensional spQUBO}
\label{sec-transform}

We observed that any two-dimensional periodic spQUBO directly corresponds to an SPIM implementation. 
We further show that any spQUBO, whether high-dimensional or non-periodic, can be rewritten as a two-dimensional periodic spQUBO, preserving the convolutional structure of the original spQUBO. 
For simplicity, in the method described below, the coefficient for each variable is omitted by setting $\amp_i=1$. 

\begin{lfigure}
\captionsetup[sub]{labelformat=empty}
\begin{center}

\includegraphics[ width=4in ]{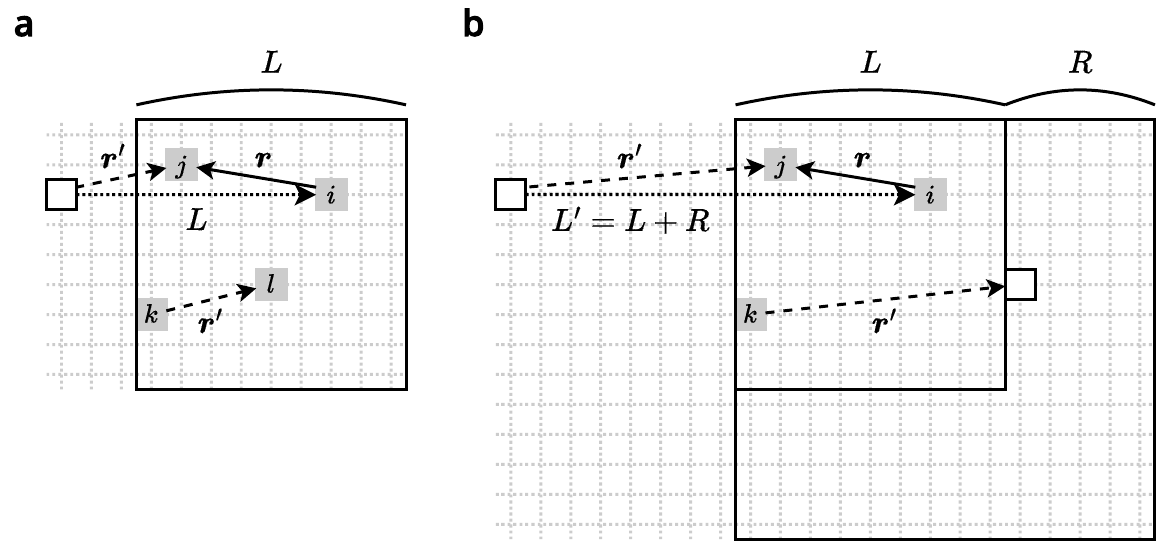}
\end{center}
\caption{
A schematic of the discussion about conflicting positions in the two-dimensional case. 
(a) The black square of size $\vlsh=(\lsh,\lsh)$ represents the configuration domain where the spins are located. 
The gray squares represent the spin locations with their indices.  
For a spin pair $(x_i, x_j)$, their relative position is $\vrel\equiv\vpos_i-\vpos_j$.  
There exists another spin pair $(x_k, x_l)$ such that the relative position $\vrel'\equiv\vpos_k-\vpos_l$ satisfies $\vrel'=(\lsh, 0)-\vrel$, where they should have the same coefficient as $J_{kl}=J_{ij}$. 
(b) The extended configuration domain with padding of size $\vran=(\ran, \ran)$. The inner black square of size $\vlsh$ represents the original configuration domain, and the outer square of size $\vlsh+\vran=\vlpa$ represents the expanded domain $\InteN{\vlpa}$. 
For the same spin pair $(x_i, x_j)$, the relative position is not greater than the padding size: $|\vrel|\le \vran$. If there is a spin pair whose relative position $\vrel'$ satisfies the same condition as above $\vrel'=(\lsh, 0)-\vrel$, either of the spins must fall outside the domain $\InteN{\vlsh}$. 
Therefore, there is no such conflicting spin pair that must have the same coefficient as $J_{ij}$. 
}
\label{fig:padding}
\end{lfigure}

First, let us examine how the periodicity condition restricts the representation power of spQUBOs. 
Figure \myref{fig:padding}a illustrates a two-dimensional spQUBO with a square configuration domain $\vlsh=(\lsh,\lsh)$, where $\vrel$ and $\vrel'$ denote the relative position between the spin pairs $(x_i, x_j)$ and $(x_k, x_l)$, respectively. 
It is apparent from the figure that $\vrel = \vpos_j - \vpos_i$ and $\vrel' = \vpos_l - \vpos_k$ represent different relative positions. 
However, if $\vrel' = \vrel + (\lsh,0)$, the periodicity condition requires that $f(\vrel)=f(\vrel')$, which is often unnatural in distance-based problems. 

This restriction can be avoided by expanding the shape of the configuration domain from $\vlsh$ to $\vlpa=\vlsh+\vran$. 
In the expanded region $\InteN{\vlpa} \setminus \InteN{\vlsh}$, which we refer to as padding, no spins are associated. 
Figure \myref{fig:padding}b shows spin positions that satisfy $\vrel' = \vrel + (\lsh+\ran, 0)$, for which $f(\vrel)=f(\vrel')$ is still required by the periodicity condition. 
However, if $\vran$ is sufficiently large, such spin positions do not occur, because for any vectors $\vpos_i, \vpos_j, \vpos_k$, and $\vpos_l$ satisfying $\vrel' = \vrel \pmod{\vlsh+\vran}$, at least one must lie within the padding region $\InteN{\vlpa}\setminus \InteN{\vlsh}$, which contains no spins. 
Thus, by introducing the padding, we can design $f$ without being restricted by the periodicity condition. 

In general, any two-dimensional spQUBO, not necessarily periodic, can be rewritten as a two-dimensional periodic spQUBO by expanding the configuration domain with $\vran=\vlsh$. 
Since this padding size is conservative, we consider reducing it based on the problem structure. 
Specifically, we assume that the interaction distance is limited as 
$|\vrel|\not\le \vran' \Rightarrow  f(\vrel)=0$, where $|\cdot|$ denotes the element-wise absolute value.
Then, we can reduce the padding size to $\vran=\vran'$, because, even if the spin positions $\vpos_i, \vpos_j, \vpos_k, \vpos_l \in \InteN{\vlsh}$ are in the above conflicting positions $\vrel' = \vrel \pmod{\vlsh+\vran}$, $f(\vrel)=f(\vrel')=0$ holds from $|\vrel|\not \le\vran$ and $|\vrel'|\not\le \vran$, and thus the periodicity condition yields no inconsistencies. 

This result is summarized as the following theorem, in which the assumption $c_i=1$ is not needed. 
The proof is provided in \addedc{Supplementary Note} \myref{sec-higherorder}. 

\begin{dfn}[locality]
\label{dfn:locality-D}
A function $f: \Inte^D \to \Real$ has a locality of $\vran \in \Nnum^D$ if, for $\vrel \in \Inte^D$, $f$ satisfies  
\AL{
|\vrel| \not\le \vran \Rightarrow f(\vrel) = 0,
}
where $|\cdot|$ denotes the element-wise absolute value. 

\end{dfn}

\begin{theorem}[Transformation of two-dimensional spQUBOs]
\label{thm:padding-spcm}
For a two-dimensional spQUBO with spatial shape $\vlsh$, whose spatial coupling function $f$ has a locality of $\vran < \vlsh$, there exists an equivalent two-dimensional periodic spQUBO with a spatial shape $\vlsh+\vran$. 

\end{theorem}
Note that from any spQUBO, we can always find a spatial coupling function with locality $\vlsh-\vect{1}$ without changing the objective function. 

More generally, we show that any higher-dimensional spQUBO can be rewritten as a two-dimensional periodic spQUBO. 

\begin{theorem}[Transformation of high-dimensional spQUBOs]
\label{thm:higherorderQUBO}
For a $D$-dimensional spQUBO with spatial shape $\vlsh$, whose spatial coupling function $f$ has a locality of $\vran < \vlsh$, there exists an equivalent two-dimensional periodic spQUBO with a spatial volume of $\ddim(\lshd+\rand)$. 

\end{theorem}
Figure~\myref{fig:mapping} illustrates how the transformation is performed. 
Summarizing the full proof given in \addedc{Supplementary Note} \myref{sec-higherorder}, this transformation is achieved by dividing the dimension indices into two disjoint sets, $\mathcal{D}_1 \cup \mathcal{D}_2 = \{1,\ldots, D\}$, such that $\mathcal{D}_1 \cap \mathcal{D}_2 = \varnothing$. We map the grid points specified by the index sets $\mathcal{D}_1$ and $\mathcal{D}_2$ to a pair of integers $\tilde{h}(\vpos) = (\tilde{h}_{\mathcal{D}_1}(\vpos_{\mathcal{D}_1}), \tilde{h}_{\mathcal{D}_2}(\vpos_{\mathcal{D}_2}))$, which defines a new grid point for the resulting two-dimensional periodic spQUBO. 
The mapping in each index set $\tilde{h}_{\mathcal{D}}$ maps to the interval $\InteN{\prod_{\idim\in\mathcal{D}}\lpad}$, where $\vlpa=\vlsh+\vran$. 
Therefore, the spatial volume of the obtained spQUBO becomes $\ddim\lpad$. 
Specifically, the $K$-dimensional coordinate vector $\vect{z} = (z_1, \ldots, z_K) = \vpos_{\mathcal{D}_1}$ is mapped to the integer value $z=(\tilde{h}(\vpos))_1$ computed as 
\AL{
z = z_1 + z_2\lpa_1 + z_3\lpa_1\lpa_2 + \cdots + z_{K} \prod_{k=1}^{K-1} \lpa_{k}. 
}
In particular, if the spatial shape is identical for all dimensions $\vlpa=(\lpa, \ldots, \lpa)$, the mapping is interpreted as the conversion between the integer number and its representation with a radix $\lpa$.
We can also transform the original spatial coupling function $f$ to $\tilde{f}$ by setting $\tilde{f}(\tilde{\vrel}) = f(\vpos_i-\vpos_j)$ for each spin pair, where $\tilde{\vrel}$ is computed from the spin position difference $\vpos_i-\vpos_j$ similarly to the spin position but considering the wrapping of $\mod{(\prod_{k=1}^{K} \lpa_{k})}$. 
For the remaining cases, we set $\tilde{f}(\vrel)=0$. 
This process for obtaining a two-dimensional periodic spQUBO enables us to implement any spQUBO on SPIM, preserving the spatially convolutional structure.

\begin{lfigure}
\captionsetup[sub]{labelformat=empty}
\begin{center}

\includegraphics[ width=6in ]{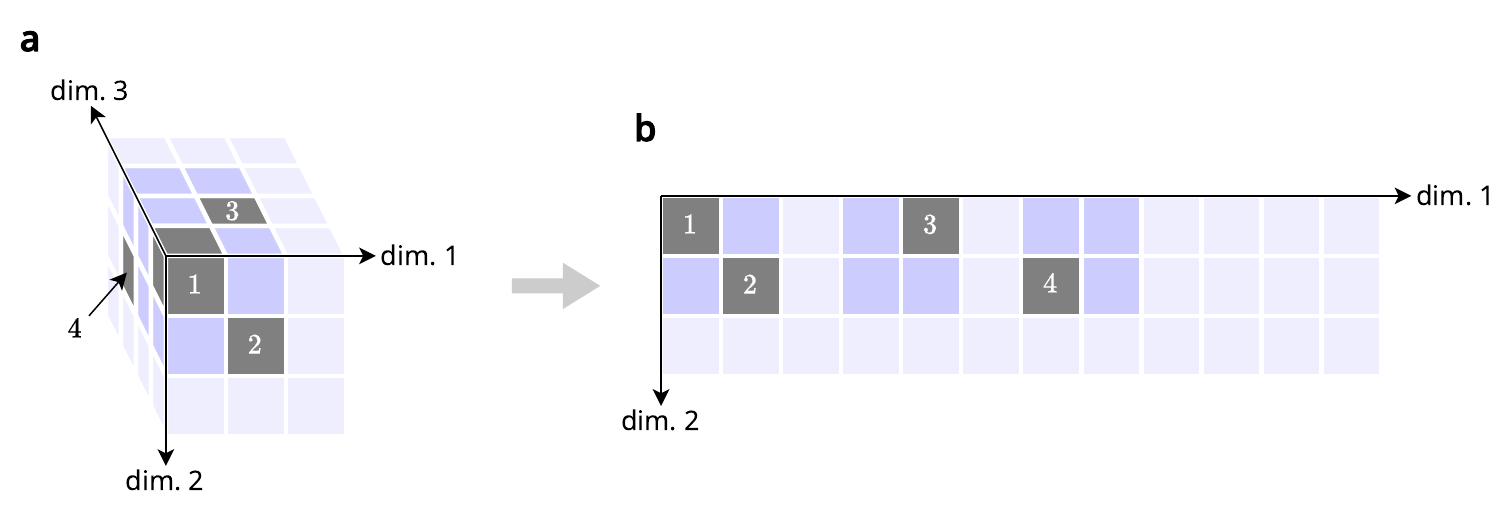}
\end{center}
\caption{
Reduction of three-dimensional spQUBO to two-dimensional. 
(a) Example of the three-dimensional spQUBO of with $(\lsh_1, \lsh_2, \lsh_3) = (2, 2, 3)$ and $(\ran_1, \ran_2, \ran_3) = (1,1,1)$. 
The problem has 4 spins located at $\vpos_1=(0,0,0)$, $\vpos_2=(1,1,0)$, $\vpos_3=(1,0,1)$ and $\vpos_4=(0,1,2)$.  
We assume symmetric couplings; $W_{ij}=W_{ji}$ for all $i, j\in \viset$. 
The darker blue area represents the original configuration domain before applying the padding. 
(b) The two-dimensional spQUBO obtained by the proposed reduction algorithm with dimension groups $\mathcal{D}_1 = \{1, 3\}$ and $\mathcal{D}_2 = \{2\}$. 
The new spin positions for the spin at $\vpos=(d_1, d_2, d_3)$ are computed as $\tilde{\vpos}=(d_1 + d_3(\lsh_1+\ran_1), d_2)=(d_1 + 3 d_3, d_2)$, and the new spatial shape is $((\lsh_1+\ran_1)(\lsh_3+\ran_3), \lsh_2+\ran_2) = (12, 3)$. 
The values of the new spatial coupling function $\tilde{f}$ are defined as 
$\tilde{f}(1,1) = \tilde{f}(-1,-1) = f(1,1,0) = W_{12}$, 
$\tilde{f}(4,0) = \tilde{f}(-4,0) = f(1,0,1) = W_{13}$, 
$\tilde{f}(6,1) = \tilde{f}(-6,-1) = f(0,1,2) = W_{14}$, 
$\tilde{f}(3,-1) = \tilde{f}(-3,1) = f(0,-1,1) = W_{23}$, 
$\tilde{f}(5,0) = \tilde{f}(-5,0) = f(-1,0,2) = W_{24}$, and
$\tilde{f}(2,1) = \tilde{f}(-2,-1) = f(-1,1,1) = W_{34}$, where $f$ is the original coupling function, and they are defined periodically with period $(12, 3)$.
}
\label{fig:mapping}
\end{lfigure}

\subsection{Numerical examples of distance-based problems}

We numerically demonstrate that distance-based problems can be implemented on the (idealized) SPIM by using concrete examples of the placement problem and the clustering problem. 
Detailed settings are provided in Methods and the full details needed to reproduce the results can be found in \addedc{Supplementary Notes} \myref{sec-placement-numerical} and \myref{sec-clustering-synthetic}, and the publicly available code repository (see Code availability). 

\begin{lfigure}
\captionsetup[sub]{labelformat=empty}
\begin{center}

\includegraphics[ width=6.5in ]{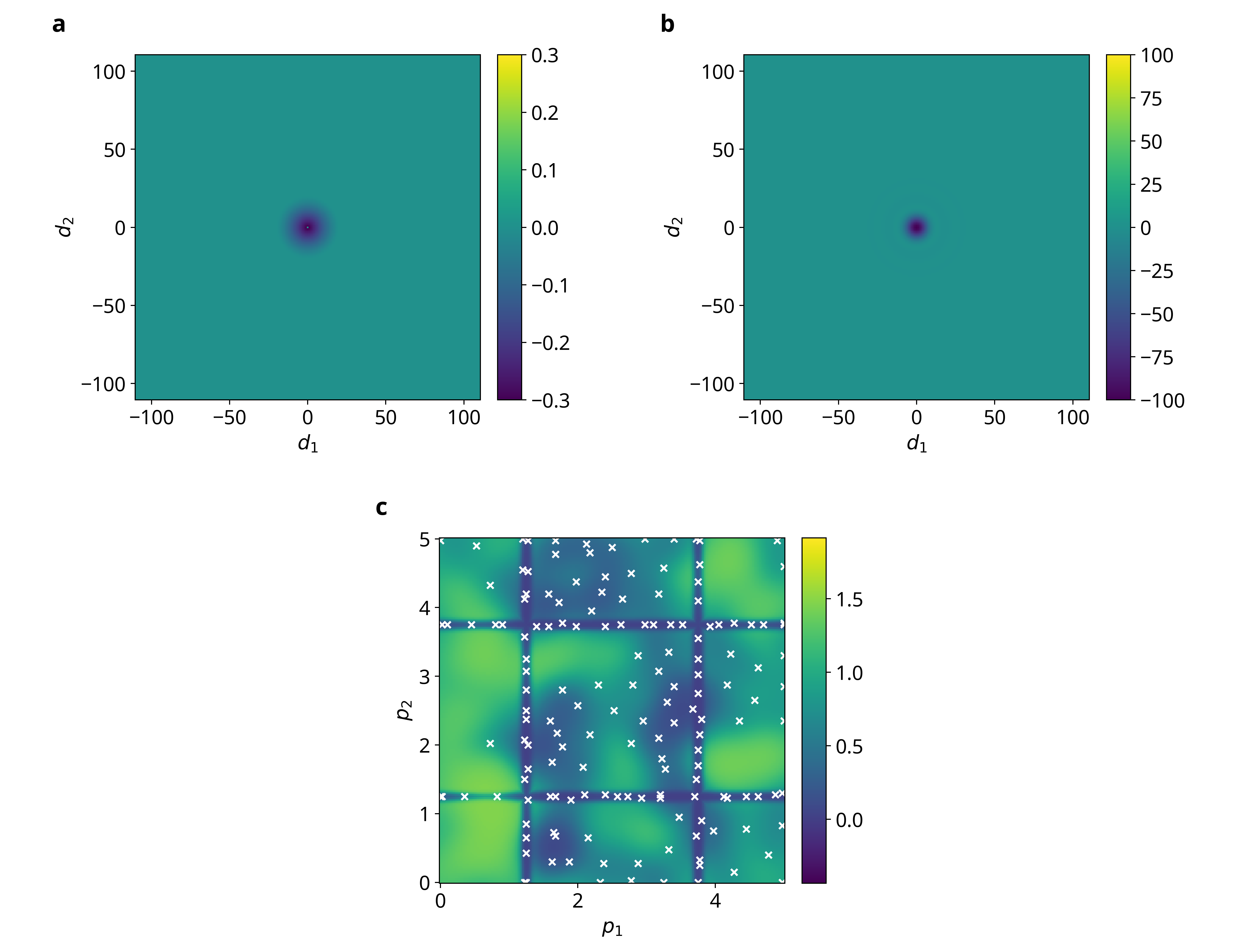}
\end{center}
\caption{
An example placement problem. 
(a) The coupling function $f$ and (b) its discrete Fourier transform $\fou{f}$ of the example placement problem.
(c) An approximate solution obtained using momentum annealing \mycite{MA} with the exact computation of the Hamiltonian values.
The color represents the placement cost at each candidate site. 
The white crosses represent the placed facilities in the solution. 
The axes for (a) and (b) represent the integer coordinates in the transformed spQUBO, and those for (c) represent the coordinates in the original problem. 
The distance-based interaction expressed in Eq.~(\ref{eq:placement_interaction_multidim}) corresponds to the center blue circle in (a) and converted to the concentric pattern in (b). 
In the solution shown in (c), the placement is limited to the blue region but distributed within it, where the placement cost is low. 
}
\label{fig:placement_qubo}
\end{lfigure}

First, we consider the placement problem with the coupling function and its discrete Fourier transform shown in Fig.~\myref{fig:placement_qubo}a,b. 
This problem can be directly represented as a two-dimensional spQUBO, where the spins are located at all integer points in $\InteN{\lsh}^2$. 
Figure \myref{fig:placement_qubo}c shows an approximate solution for the spQUBO obtained by momentum annealing \mycite{MA} on a conventional laptop computer. 
The solution places many facilities on the points where the placement cost is low, and these facilities are appropriately spaced to avoid diminishing the utility gain. 

\begin{lfigure}
\captionsetup[sub]{labelformat=empty}
\begin{center}

\includegraphics[ width=6.5in ]{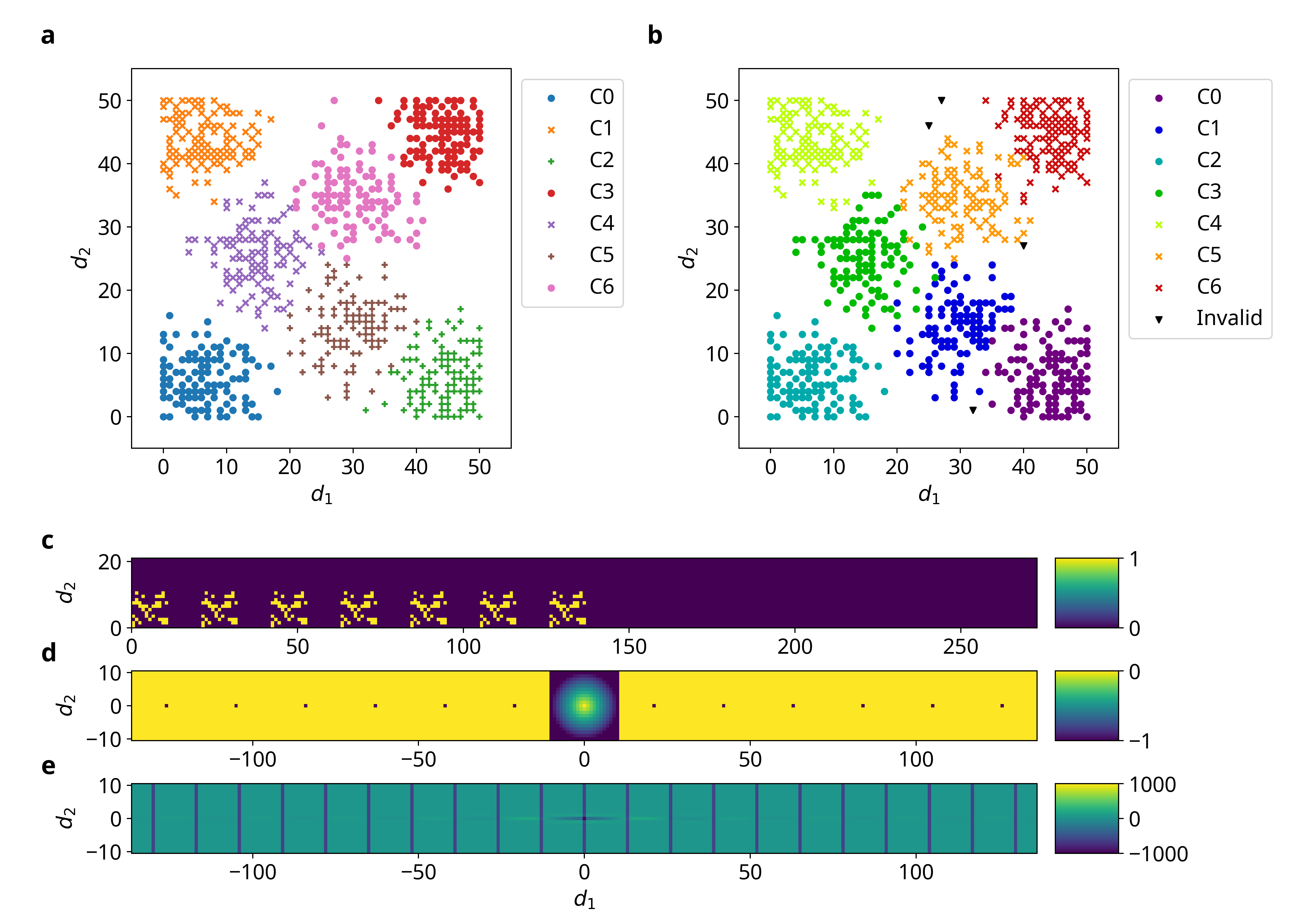}
\end{center}
\caption{
(a,b) Example of clustering problem of dimension $D=2$.
The distribution of the points and clusters (a) at generation and (b) in the approximate solution obtained using momentum annealing \mycite{MA} with the exact computation of the Hamiltonian values. 
The clusters are represented by the shape and color of the points. 
The spins represented by black triangles denote the spins with invalid output where the exact-one constraint is violated. 
The axes represent the integer coordinates of the problem. 
(c-e) The transformed two-dimensional spQUBO for an example of smaller sized problem.
The axes represent the integer coordinates of the transformed spQUBO. 
(c) The spin arrangement, where the yellow dots indicate the locations where spins are mapped.
(d) The coupling function $f$ and (e) its discrete Fourier transform $\fou{f}$.
The axes for (a) and (b) represent the coordinates in the original problem, and those for (c-e) represent the integer coordinates in the transformed spQUBO. 
In the panel (c), the points in the original clustering problem are copied $K$ times corresponding to the possibility of the cluster assignments. 
The distance-based cost expressed in Eq.~(\ref{eq:clustering_HAg_multidim}) corresponds to the yellow circle at the center of (d) and the horizontal pattern in (e). 
On the other hand, the assignment constraints expressed in Eq.~(\ref{eq:clustering_HB_multidim}) corresponds to horizontally arranged blue dots in (d) and vertical lines in (e). 
}
\label{fig:clustering}
\end{lfigure}

Next, we consider a clustering problem with $K = 7$ clusters, where the $M$ data points are located in $\InteN{\lsh}^2$ for $\lsh = 51$, as shown in Fig.~\myref{fig:clustering}a. 
This problem is formulated as a three-dimensional spQUBO with $MK$ spins in the configuration domain of size $\lsh^2K \approx 1.8\times 10^4$. 
Figure \myref{fig:clustering}b shows an approximate solution for the spQUBO obtained in the same way as above, where clusters are qualitatively well identified. 
In this process, we reduced the problem to a two-dimensional spQUBO using the proposed reduction algorithm in Theorem \myref{thm:higherorderQUBO}. 
Figures \myref{fig:clustering}c-e represent the reduced two-dimensional spQUBO for a smaller-sized example problem.  

These results indicate that the Hamiltonians for the placement and clustering problems can be computed using single-shot SPIMs. 
Since the optimization algorithms are independent of Hamiltonian computation, we can employ alternative optimization algorithms to improve the solution quality in more complex problems. 

\subsection{Scaling and comparisons to other mapping schemes}
\label{sec-benchmark}

Different approaches have been proposed to improve the representation power of SPIMs. 
For example, Ref.~\myciten{Arbitrary25} represents the value of a product of interacting spins, such as $x_ix_j$ as a new spin in SPIM. 
Another approach, as in Ref.~\myciten{FocalDivision25}, copies each spin on the SLM as many times as the number of components in the multiplexed Mattis model. 
We also described a generic method to convert an arbitrary QUBO problem into an spQUBO using one-hot vectors. 
Among these approaches, the spQUBO formulation has an advantage in its ability to preserve the spatial structure of the problem and adapt it to the optical nature of SPIMs. 

To assess the benefits of exploiting the spatial structure, we compare the scaling of the required SLM size of our spQUBO approach with those of spatial multiplexing and spin-pair mapping, as shown in Table~\myref{tab-comparison}.
For a $D$-dimensional spQUBO with $N$ variables and configuration domain $\InteN{\vlsh}$, the spatial volume is $V=|\InteN{\vlsh}|=\dprod \lsh_d$. 
We define a packing density of the variables as the ratio of their number to the spatial volume.  
We assume that the problem size scales while maintaining the packing density, denoted by $\alpha=N/V$.  
We also assume it scales also while maintaining the ratio of the locality (Def. 3) to the spatial shape; that is, there exists a constant $\beta$ such that $\vran\approx\beta\vlsh$. 
Then the spatial volume $V$ depending on the number of variables $N$ is estimated to scale up to $(1+\beta)^D\alpha^{-1}N$. 

When the spatial multiplexing as in Ref.~\myciten{FocalDivision25} is applied with a Mattis-type rank-$k$ problem, the spins can be packed in each row without regard to their positions.
Hence, the required SLM size equals the number of spins $kN$. 
Of course, in cases such as the number partitioning problem \mycite{PierangeliPRL} and the knapsack problem \mycite{SPBM}, where $k$ is a small constant, the scaling is linear in $N$.
However, for generic full-rank problems, the spatial volume increases to $N^2$. 

When using the spin pair mapping \mycite{Arbitrary25}, we do not need to implement the product $x_ix_j$ for the spin pairs that have no interaction (i.e. $W_{ij}=0$). 
Therefore, the required SLM size, which equals the number of spins, is reduced to the number of non-zero interactions, denoted by $E$. 
Under the assumption of uniformly distributed spins and a small locality parameter $\beta$, the number of such non-zero elements can be estimated as $E\approx (2\beta)^DN^2/2$. 

For fixed parameters $\alpha$, $\beta$, and $D$, the spatial volume scales better than the other mapping schemes, as summarized in Table~\myref{tab-comparison}. 
We can also observe that these leading terms are governed by the locality parameter $\beta$, which reflects the range of interactions in spQUBO. 
As $\beta$ increases, the spatial volume of the proposed method grows more slowly than the spin pair mapping. 
For example, in the extreme case of $\vran=\vect{1}$ (i.e., $\beta\approx 0$), as in the image restoration problem with nearest neighbor interactions, spin pair mapping is advantageous. 
On the other hand, in the case of $\vran=\vlsh$ (i.e., $\beta=1$), as in the clustering problem, the proposed method becomes more efficient. 

In summary, spQUBO is suited to problems with spatial structure, whereas spatial multiplexing \mycite{FocalDivision25} is suitable for low-rank problems and spin-pair mapping \mycite{Arbitrary25} is suitable for sparse problems.
Table \myref{tab-comparison} compiles the spatial volume estimates discussed above, along with representative values for instances of the placement and clustering problems.  
From the viewpoint of spatial volume, we see that by exploiting spatial structure the proposed method can efficiently ``pack'' spins and interactions into the SLM and the observation plane of the SPIM. 
It should be noted that, for the clustering problem, where sparsity is not indicated by the locality $\beta=1$, there is an additional sparsity that is not attributed to locality. 
As described in the Method section, we can estimate the reduced spatial volume using this sparsity as $V = E \approx (\lvert\mathcal{N}\rvert^2 K + K^2 \lvert\mathcal{N}\rvert) / 2$${}\approx 1.7\times 10^6,$ which grows faster than linearly in the number of spins and remains larger than the proposed method in Table \myref{tab-comparison}. 

\begin{table*}
\centering
\begin{tabular}{cccc}
& Required SLM size & Placement problem & Clustering problem \\\hline
SpQUBO & 
$\approx (1+\beta)^D\alpha^{-1}N$         & $\approx 4.8\times 10^4$ & $\approx 1.4\times 10^5$ \\
Spatial multiplexing \mycite{FocalDivision25} & 
$kN$                              &       $= 1.6 \times 10^9$ &      $= 2.4\times 10^7$ \\
Spin pair mapping \mycite{Arbitrary25} & 
$E $ & $\approx (2\beta)^D N^2 / 2 \approx 3.2\times 10^7$ &       $\approx N^2/2 = 1.2\times 10^7$ 
\end{tabular}
\caption{Comparison of the spatial volumes scaling (Column 2) between the QUBO mapping schemes. $E$ denotes the number of non-zero interactions ($W_{ij}\neq 0$) in the problem. Columns 3 and 4 show the values for example cases for the placement problem and the clustering problem: they are calculated with $N=4.0\times 10^4, D=2, \alpha=1, \beta\approx 0.1, k=N$ for the placement problem, and with $N = |\mathcal{N}|K = 4900, D=3, \alpha = 700/2500, \beta=1, k=N$ for the clustering problem. These parameter values roughly corresponds to the setups of the numerical examples shown in Figs.~\ref{fig:placement_qubo} and \ref{fig:clustering}a,b, as described in \addedc{Supplementary Notes} \myref{sec-placement-numerical} and \myref{sec-clustering-synthetic}.}
\label{tab-comparison}
\end{table*}
\subsection{DFT-based computation of spQUBO}
\label{sec-beyond}

As discussed so far, the formulation of spQUBO is motivated by efficient optical computation on the SPIM. 
However, the convolutional structure of spQUBO also allows efficient numerical computation using fast Fourier transforms. 

Specifically, the Hamiltonian of spQUBO can be computed using DFT as stated in the following theorem, whose proof is given in \addedc{Supplementary Note} \myref{sec-proof-fft}. 

\begin{theorem}

\label{thm:fft-hamiltonian}
Let $F$ be the Hamiltonian of a two-dimensional periodic spQUBO, without the bias term for simplicity, expressed as follows: 
\AL{
F = \frac{1}{2}\sumij c_ic_jf(\vdi-\vdj)x_ix_j. 
}
We define a function $\xi\colon\Inte^2\to\Real$ with period $\vlsh$ as 
\AL{
\xi(\vx) = \begin{cases}
c_ix_i & \text{if } \vx = \vdi, \\ 
0 & \text{otherwise}
\end{cases}
}
for $\vx\in\InteN{\vlsh}$. 
Then, it holds that 
\AL{
F = \frac{1}{2V}\sum_{\spos\in\InteN{\vlsh}}\|\fou{\xi}(\spos)\|^2\fou{f}(\spos), 
}
where $V=\lsh_1\lsh_2$ is the spatial volume and $\mathcal{F}$ denotes the discrete Fourier transform. 

\end{theorem}
In many Ising solvers implemented on digital computers, the interactions between variables are computed via matrix-vector product (MVP) of the form $J\vx$, whereas the SPIM computes the Hamiltonian $\vx\tr J \vx$ directly (see \addedc{Supplementary Note}  
\myref{sec-mvms-in-solver} for a brief review of MVPs appearing in Ising solvers).
The convolutional structure of $J$ in spQUBO enables efficient computation of the MVP using DFT. 
Specifically, for a two-dimensional periodic spQUBO, the MVP can be computed using DFT, as stated in the following theorem. 
The proof is given in \addedc{Supplementary Note} \myref{sec-proof-fft}. 

\begin{theorem}

\label{thm:fft-mvm}
Let $a_i$ be the $i$-th element of the MVP for a two-dimensional periodic spQUBO as follows: 
\AL{
a_i = \sumj c_ic_jf(\vdi-\vdj)x_j. 
}
We define a function $\xi\colon\Inte^2\to\Real$ with period $\vlsh$ as in Theorem \myref{thm:fft-hamiltonian}. 
Then, it holds that 
\AL{
a_i = c_i \ifouAS{\fou{\xi}\odot\fou{f}}(\vdi), 
}
where $\odot$ denotes the Hadamard element-wise product. 

\end{theorem}
Let us estimate the computational cost of the DFT-based computation and direct computation of the Hamiltonian. 
Note that we ignore the cost of transforming the problem into a two-dimensional periodic spQUBO, which needs to be performed only once for an spQUBO instance. 

We can use the algorithm of Fast Fourier Transforms (FFT) for DFT, so the cost of the DFT-based computation can be estimated as 
\AL{
O(V \log V), 
}
where $V$ is the spatial volume that estimated in the previous section.  
This scaling is advantageous over the direct computations, whose cost grows at the rate $O(N^2)$.  

To demonstrate the efficiency of the DFT-based MVP computation for spQUBO, we conduct a numerical comparison between the DFT-based and direct MVP computations. 
We solved the placement problem using momentum annealing (MA), which requires an MVP in each iteration. 
Figure \myref{fig:fig2} shows the computation time for each problem size $N$, measured for the same number of iterations, with both methods yielding identical results. 
The DFT-based computation shows better scaling with respect to $N$ than the direct computation. 
The numerical experiments were conducted using Python with the NumPy package on an M1-based MacBook. 

Optical MVP computation based on the convolutional structure of spQUBO is a promising direction for future research. 
While this paper focuses on the SPIM architecture, such optical MVP implementations could be integrated with other Ising solvers to achieve efficient optimization of spQUBOs. 

\begin{sfigure}
\captionsetup[sub]{labelformat=empty}
\begin{center}

\includegraphics[ width=3in ]{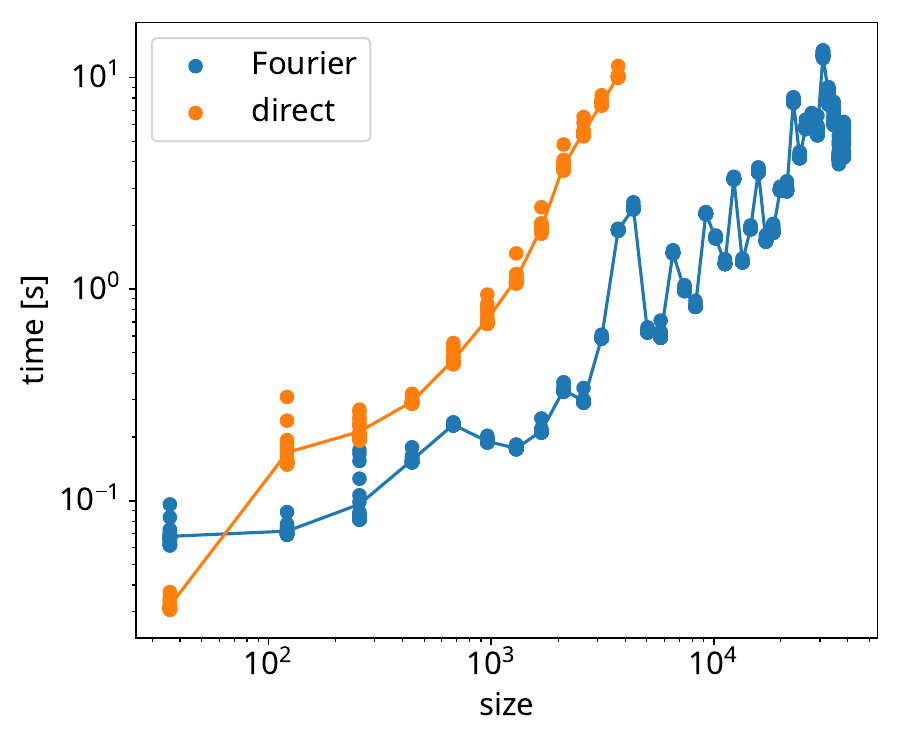}
\end{center}
\caption{
The computation time for a step of MA with FFT-based and direct MVPs for the placement problem. 
The horizontal axis and the vertical axis represent the number of spins $N$ and the required time, respectively. 
There are 25 runs for each $N$ and each result is represented as a point. 
The curve represents the average of the 25 runs for each $N$. 
The DFT-based computation shows better scaling with respect to the number of spins than the direct computation. 
}
\label{fig:fig2}
\end{sfigure}

\section{Discussion}
\label{sec-discussion}

We have clarified the single-shot representation power of the idealized SPIM by introducing the spQUBO formulation. 
Optical implementations of SPIMs are inevitably constrained by physical limitations, such as aberration, system alignment, the quality of the coherent beam, and the photon budget required to maintain an acceptable error for the Hamiltonian computation. 
In this study, we nevertheless idealized the system focusing on its spatial volume, which is constrained by the system's principal parameters: wavelength $\lambda$, focal length $f$, and the pitches $\ell$ and $a$ of the SLM and observation, respectively. 
We mainly studied theoretical scaling in terms of the spatial volume and scarcely mentioned specific optical implementations, which are rapidly developing. 
Especially, whether the physically realized SPIM system still exhibits superiority over other systems even when other aspects, such as the required energy or photon budget for solving convolutional dense problems, are taken into account remains an open issue. 

More specifically, to evaluate effectiveness, it is important to analyze its performance and efficiency in the overall problem-solving process. 
We should consider the desired accuracy level of Hamiltonian computation, or equivalently the resources required to achieve it, in relation to the target solution quality and the characteristics of the search algorithms that run on SPIM's Hamiltonian computation. 
This analysis should incorporate various factors such as the energy landscape of the problem and noise schedule for simulated annealing. 
Our theoretical results advance our understanding of the class of combinatorial optimization problems for which we can fully leverage the expected superior scalability and efficiency of SPIMs, and the proposed framework can be used as a basis of the development and the evaluation of the future optical SPIM implementations. 
In this sense, the spQUBO formulation provides a foundational step toward enhancing the applicability of SPIMs. 

The superiority of SPIM for the distance-based problems has also been suggested in prior work. 
For example, in the original work of SPIM \mycite{PierangeliPRL}, it was already suggested that the coupling coefficients become distance-based in general. 
Moreover, Ref.~\myciten{Correlation21} proposes computing a correlation function to obtain such distance-based Hamiltonian value in SPIM, which is essentially identical to the way the SPIM is modeled in ths study. 
The distance-based problems discussed in this paper can also be referred to as translation invariant problems. 
The model whose couplings are represented by a circulant matrix is translation invariant and can likewise be implemented \mycite{Circulant25}.  
Moreover, low-dimensional translation invariant spin glasses can be implemented on SPIM using the correlation function method \mycite{Circulant25}, and the Fourier-mask SPIM \mycite{PhotonicSG25} uses a Fourier transformation to design the target image, similarly to our approach. 
However, these are only for specific cases such as spin glasses in two- or three-dimensional grids that have sparse interactions in nearest or next-nearest neighbors.  
This study provides a generalization of these approach to arbitrary dense or high-dimensional problems, thereby offering a more comprehensive and systematic way of implementing them on SPIM and enabling their application to practical problem instances. 

In performance evaluation of Ising solvers, especially when using large-scale problems, it is common to use problems with sparse coupling matrices, such as the maximum cut problem on sparse networks\mycite{Hamerly19,Bohm21}. 
As already discussed, this sparsity can be used for efficiently representing QUBO problems on SPIM architecture \mycite{Arbitrary25}. 
However, real-world problems are not necessarily sparse, and our results suggest that SPIMs also have a unique advantage when handling dense coupling matrices with convolutional structures. 
Therefore, it is important to evaluate the problem-specific performance of Ising solvers using a wider variety of problem instances, including those with dense interactions. 
Such coupling-wise operations can also be used for simulation of spin models with many-body interactions by Ising models by adding extra spins for each interaction \mycite{UniversalIsing16}. 
A universality of the Ising model can be shown through this reduction \mycite{UniversalIsing16}, whereas no physical constraints are imposed, unlike those we considered for the spQUBO and the SPIM in this study. 
Studying to what extent of this universality is preserved under these physical constraints is also an interesting theoretical direction for future research. 

Statistical machine learning of spQUBO Hamiltonians is a promising approach to fully exploit the unique efficiency of the spQUBO formulation in handling dense interactions. 
Several studies have investigated the use of Ising solvers for black-box optimization \mycite{BBOandDFO,BBOTwoDecades,BBOforAutomated} by formulating unknown objective functions as Ising Hamiltonians or the equivalent quadratic functions \mycite{BOCS18,Walsh19,COMEX20,QUBOmetamaterial,Mercer21,QUBOphasefield,QUBOresonance}. 
The spQUBO formulation can be advantageous in learning dense interactions due to its small number of degrees of freedom, which is at most the spatial volume $O(V)$, whereas the direct representation of coupling matrices requires $O(N^2)$ parameters. 
The spQUBO is also expected to be effective for black-box optimization problems with intrinsically convolutional structures such as those embedded in two- or three-dimensional physical spaces. 
For example, to obtain an optimal design of the mounting holes of a printed circuit board, the design variables are assigned in the two-dimensional space of the circuit board \mycite{QUBOresonance}. 
Also, for optimization of the traffic light patterns in an urban area, the design variables are assigned to traffic lights placed on a two-dimensional area, which is two-dimensional \mycite{Inoue21,Traffic24}. 

Multiplexing the spQUBO is an interesting direction for future research. 
Several studies have explored extending the SPIM to represent higher-rank interactions by multiplexing the SPIM in various modes \mycite{SPIMQuadrature,SPIMWavelength,SPIMSakabe,SPBM,SPIMYe2023v2,OguraCREST24,FocalDivision25,Circulant25}. 
These multiplexing approaches can be combined with the spQUBO, leading to more efficient representations of QUBOs using multiple spQUBOs. 

To conclude, we have clarified the class of QUBOs that can be efficiently represented by the idealized SPIM without multiplexing. 
The proposed class, spQUBO, is capable of representing Ising Hamiltonians with dense interactions and convolutional structures. 
Based on the spQUBO formulation, we expect further progress in both the implementation and application of the SPIM, which can exhibit superior efficiency and scalability. 

\section{Methods}
\subsection{SPIM}
\label{sec-spim-appx}

Here we recapitulate the SPIM and derive its Hamiltonian based on the original paper \mycite{PierangeliPRL}, in our notation. 
Let $i$ be the index for the elements of the SLM arranged in a two-dimensional grid; let us assume that, for each $i$, there exists an integer point $\vpos_i\in\Inte^2$ such that the coordinates of the $i$-th element are represented as $\spos_i=2\ssize \vpos_i$. 
The SLM is configured to modulate the incoming laser light so that the electric field on the SLM plane is represented as 
\AL{
\tilde{E}(\spos)=\sumi \xi_i \sigma_i \tilde\delta_\ssize (\spos-\spos_i),
}
where $\xi_i$ and $\sigma_i$ correspond to the amplitude and the phase on the $i$-th element, and $\tilde{\delta}_\ssize (\spos)$ is the indicator function that represents the shape of each SLM element, which is the square of size $\ssize $ centered on $(0,0)$. 

After propagating the modulated light through the optical system with lens as shown in Fig.~\myref{fig:SLM}, we obtain the inverse Fourier transform of $\tilde{E}(\spos)$ as the electric field $E(\vect{x})$ on the observation plane. 
In particular, with the optical setup of the focal length $\flen$ of the lens and the wavelength $\lambda$ of the laser, the obtained inverse Fourier transform $f(\vx)$ for the $i$-th element $\tilde{f}(\spos)=\tilde{\delta}_\ssize (\spos-\spos_i)$, is represented by using $\vect{y}=\spos-\spos_i$ as 
\AL{
f(\vect{x})&=\int_{-\infty}^{\infty}d\spos \tilde{\delta}_\ssize (\spos-\spos_i)e^{(2\pi \iunit/\flen\lambda) \vect{x}\tr\spos}\NN 
&=\int_{-\ssize }^{\ssize }dy_1\int_{-\ssize }^{\ssize }dy_2 e^{(2\pi \iunit/\flen\lambda)\vx\tr(\vy+\spos_i)}\NN
&=e^{(2\pi \iunit/\flen\lambda)\vx\tr \spos_i}\int_{-\ssize }^{\ssize }dy_1\int_{-\ssize }^{\ssize }dy_2e^{(2\pi \iunit/\flen\lambda)\vx\tr \vy}\NN
&=e^{(4\pi \iunit \ssize/\flen\lambda) \vect{x}\tr \vpos_i}\delta_\ssize (\vect{x}),
}
where $\delta_\ssize (\vect{x})$ is the inverse Fourier transform of $\tilde\delta_\ssize (\vect{k})$. 
Therefore, from the linearity of the Fourier transform, we can represent the electric field on the observation plane as 
\AL{
E(\vect{x})=\sumi \xi_i \sigma_i\delta_\ssize (\vect{x})e^{(4\pi \iunit \ssize/\flen\lambda)  \vect{x}\tr\vpos_i}. 
}
We can observe its intensity $\inten(\vect{x})=|E(\vect{x})|^2=E(\vect{x})\overline{E(\vect{x})}$ by an image sensor, which can be computed as 
\AL{
\inten(\vect{x}) = \sumij \xi_i\xi_j \sigma_i\sigma_j(\delta_\ssize (\vect{x}))^2e^{(4\pi \iunit \ssize/\flen\lambda)  \vect{x}\tr(\vpos_i-\vpos_j)},
}
where we denote the complex conjugates by $\overline{(\cdot)}$. 

If we observe only the intensity at the origin $\inten(\vect{0})$, we obtain 
\AL{
I(\vect{0}) = \sumij \xi_i\xi_j \sigma_i\sigma_j. 
}
Using this observation to define the Hamiltonian as $H=-(1/2)I(\vect{0})$ yields the rank-one coupling matrix $J$ defined as $J_{ij}=\xi_i\xi_j$ mentioned in the text and other literature \mycite{SPIMQuadrature,SPIMWavelength,SPIMSakabe,SPBM,SPIMYe2023v2,OguraCREST24,FocalDivision25,Circulant25}. 

Let us generalize the observation as proposed in Ref.~\myciten{Correlation21} by defining a reference function $\refinten(\vx)$ so that the observation is proportional to its inner product to the image: 
\AL{
\hamil&\propto-\frac{1}{2}\int d\vect{x}\refinten(\vect{x}) \inten(\vect{x})\NN
&=-\frac{1}{2}\int d\vect{x} \refinten(\vect{x}) \sumij \xi_i\xi_j \sigma_i\sigma_j(\delta_\ssize (\vect{x}))^2e^{(4\pi \iunit \ssize/\flen\lambda)  \vect{x}\tr(\vpos_i-\vpos_j)}
}
If the SLM's element size $\ssize $ is small so that $\tilde{\delta}_\ssize (\spos)$ is close to Dirac delta function and its inverse Fourier transform is approximated as $\delta_\ssize (\vect{x})\approx 1$, we obtain the approximate equation 
\AL{
\hamil&\propto-\frac{1}{2}\int d\vect{x} \refinten(\vect{x}) \sumij \xi_i\xi_j \sigma_i\sigma_j e^{(4\pi \iunit \ssize/\flen\lambda)  \vect{x}\tr(\vpos_i-\vpos_j)}
}
Because the pitch size in this setting is $\ell=2\ssize$, we obtain 
\AL{
\hamil&\propto -\frac{1}{2}\int d\vect{x} \refinten(\vect{x}) \sumij \xi_i\xi_j \sigma_i\sigma_j e^{(2\pi \iunit\ell/\flen\lambda)  \vect{x}\tr(\vpos_i-\vpos_j)}.
}
Under the assumption of symmetric reference $\refinten$, we obtain the Hamiltonian in the form of Eq.~(\myref{eq:SPIM-hamiltonian-sym}). 

As noted in Ref.~\myciten{PierangeliPRL}, minimizing $\hamil$ can also be interpreted as minimizing the difference between the functions 
\AL{
&\int (\refinten(\vect{x})-\inten(\vect{x}))^2d\vect{x} \NN
&=\int\refinten(\vect{x})^2 d\vect{x} - 2\int \refinten(\vect{x})\inten(\vect{x})d\vect{x}+\int \inten(\vect{x})^2d\vect{x},
}
assuming $\int \refinten(\vx)^2\approx \int \inten(\vx)^2$ are constant. 

\subsection{Formulation of placement problem as QUBO}
\label{sec-placement-derivation}

Let us consider a placement problem to determine the optimal placement of facilities on a two-dimensional plane. 
The objective is to maximize the total utility minus the total cost of the placements. 
Specifically, we consider a placement problem on grid points $\InteN{\lsh}^2$, where $\lsh$ is the size of the grid. 
All the points in $\InteN{\lsh}^2$ are indexed by $i$, and $\vpos_i$ denotes the $i$-th grid point. 
Placements are represented by binary variables $x_i \in \{0, 1\}$, 
where $x_i = 1$ if a facility is placed at position $\vpos_i\in\InteN{\lsh}^2$, and $x_i = 0$ otherwise. 

We assume a diminishing marginal utility, defined at each point $\vz$ as 
\AL{
\xi_\vz = \xi(n_\vz),
}
where a quadratic function 
\AL{
\xi(n_\vz) = -an_\vz^2 + bn_\vz, \LB{eq:placement-quadratic-utility}
}
with parameters $a, b > 0$, is applied to the number of facilities $n_\vz$ placed within a distance $R$ from $\vz$. 

To obtain overall utility, we integrate $\xi(n_{\vz})$ over $\vz$: 
\AL{
\Xi=\int_\vz \xi(n_{\vz}) d\vz
}
The number of facilities within a radius $R$ of $\vz$, denoted by $n_{\vz}$, is 
\AL{
n_{\vz}=\sumi  d_i(\vz)x_i, 
}
where $d_i(\vz)\equiv\indi(\|\vpos_i-\vz\|\le R)$ is the indicator function of the circle with radius $R$ centered at $\vpos_i$. 
By substituting it into $\xi(n_{\vz})$, we obtain 
\AL{
\xi(n_{\vz})=-a\sumij d_i(\vz)d_j(\vz)x_ix_j+b\sumi  d_i(\vz)x_i,
}
and, by integration, we obtain 
\AL{
\Xi=-a\sumij W'_{ij}x_ix_j+\sumi bb'_ix_i,
}
where $W'_{ij}=\int_\vz d_i(\vz)d_j(\vz)$ and $b'_i=\int_\vz d_i(\vz)$. 
We can compute the coefficients as 
\AL{
W'_{ij}&=g(\rel_{ij})=\begin{cases} 
2(R^2\theta_{ij} - \rel_{ij}^2\tan\theta_{ij}/4)&(\rel_{ij}<2R)\\
0&(\rel_{ij}\ge 2R)
\end{cases},\LN{eq:placement_interaction_multidim}
b'_i&=S, 
}
where $S=\pi R^2$ is the area of a circle with radius $R$ and $\theta_{ij}=\cos^{-1} (\rel_{ij}/2R)$. 

Thus, we obtain the quadratic form of $\Xi$ as 
\AL{
\Xi&=-a\sumij g(\rel_{ij}) x_ix_j+ bS\sumi  x_i.
}

The placement cost at the $i$-th grid point is assumed to be given as $c_i$. 
Consequently, the objective function of the placement problem is obtained by subtracting the placement costs from $\Xi$ as 
\AL{
\obj&=\Xi-\sumi \plcost_ix_i\\
&=-a\sumij g(\rel_{ij})x_ix_j + \sumi  (bS-\plcost_i)x_i.
}
An example of the spatial coupling function of the obtained spQUBO is shown in Fig.~\myref{fig:placement_qubo}a. 

\subsection{Formulation of clustering problem as QUBO}
\label{sec-clustering-derivation}

For the clustering problem with $K=2$ clusters, the partitioning is represented by binary variables $x_i\in\{0, 1\}$, which indicate the cluster index to which the $i$-th data point is assigned. 
The total distances $\hamil_g$ within the $g$-th cluster for each $g\in\{0, 1\}$ are given by 
\AL{
\hamil_0 &= \sumij \rel_{ij}(1-x_i)(1-x_j),\\
\hamil_1 &= \sumij \rel_{ij}x_ix_j,
}
where $\rel_{ij}$ denotes the distance between the $i$-th and $j$-th data points. 
The total distance to be minimized is given by 
\AL{
\hamil &= \hamil_0 + \hamil_1 \NN
&= 2\Bigg(\sumij \rel_{ij}x_ix_j\Bigg) - 2\Bigg(\sumi x_i\sumj \rel_{ij}\Bigg) + \Bigg(\sumij \rel_{ij}\Bigg).
}
By defining the objective function as $\obj=-\hamil/2$, we obtain a QUBO with coupling matrix and bias vector given by 
\AL{
W_{ij} = -2\rel_{ij}, \quad b_i = \sumj \rel_{ij}. 
}
Furthermore, if the data points $\vpos_i$ are located at grid points in $D$-dimensional space 
and the distances are given by $\rel_{ij} = \|\vpos_i-\vpos_j\|$, 
the clustering problem becomes a $D$-dimensional spQUBO. 

For the general case of $K>2$ clusters, the partitioning is represented by binary variables $x_{i\gidx}\in\{0, 1\}$, which indicate whether the $i$-th data point is assigned to the $g$-th cluster ($x_{i\gidx}=1$) or not ($x_{i\gidx}=0$). 
The variable $x_{i\gidx}$ is associated with the point $\tilde{\vpos}_{i\gidx} = (\vpos_i, g)$ in $(D+1)$-dimensional space, where $\vpos_i$ is the $i$-th data point in $D$-dimensional space, and $g$ is the cluster index. 
We define the objective function to be minimized as $\hamil=\hamil_\mathrm{A}+C\hamil_\mathrm{B}$, where 
\AL{
\hamil_\mathrm{A}&=\sum_{\gidx=1}^{K} \hamil_{\mathrm{A},\gidx},\\
\hamil_{\mathrm{A},\gidx}&=\sumij \rel_{ij}x_{i\gidx}x_{j\gidx},\LN{eq:clustering_HAg_multidim}
\hamil_\mathrm{B}&=\sumi\left(1-\sum_{\gidx=1}^{K}x_{i\gidx}\right)^2,\LB{eq:clustering_HB_multidim}
}
$C$ is the weight parameter and $\rel_{ij}=\|\vpos_i-\vpos_j\|$ represents the distance between the two points $\vpos_i$ and $\vpos_j$.
As in the case of $K=2$, $\hamil_{\mathrm{A},\gidx}$ represents the sum of the distances of all point pairs in cluster $g$. 
Thus, minimizing $\hamil_\mathrm{A}$ works for assigning the points in the neighborhood to the same cluster. 
On the other hand, $\hamil_\mathrm{B}$ represents the constraint that each point is assigned to exactly one cluster. 
By expanding $\hamil_\mathrm{A}$, we have 
\AL{
\hamil_\mathrm{A}&=\sum_{\gidx=1}^{K} \sumi \sumj \rel_{ij}x_{i\gidx}x_{j\gidx}\NN
&=\sumi\sumj\sum_{\gidx=1}^{K} \sum_{\gidx'=1}^{K} \rel_{ij}\delta_{\gidx, \gidx'}x_{i\gidx}x_{j\gidx'},\LB{eq:clustering_HA}
}
and, by expanding $\hamil_\mathrm{B}$, we have 
\AL{
\hamil_\mathrm{B}&=\sumi\left(1-\sum_{\gidx=1}^{K}x_{i\gidx}\right)^2\NN
&=\AS{\sumi\sumj\sum_{\gidx=1}^{K}\sum_{\gidx'=1}^{K}\delta_{i,j}x_{i\gidx}x_{j\gidx'}}-2\AS{\sumi\sum_{\gidx=1}^{K}x_{i\gidx}}+N.\LB{eq:clustering_HB}
}
Therefore, the problem is a QUBO to maximize $\obj=-\hamil$, where the coupling coefficient between $x_{i\gidx}$ and $x_{j\gidx'}$ is represented as 
\AL{
W_{i\gidx,j\gidx'}=-2(\rel_{ij}\delta_{\gidx,\gidx'}+C\delta_{i,j}). 
}
The coefficients are represented in the form of a spatial coupling function as 
\AL{
W_{i\gidx,i'\gidx'} &= f(\tilde{\vpos}_{i\gidx}-\tilde{\vpos}_{i'g'})\NN 
&= -2(\|\vpos_i-\vpos_{i'}\| \delta(g-g') + C\delta(\vpos_i-\vpos_{i'})), \LB{eq:clustering-J}
}
which is based on the relative positions of spins. 
The biases are computed as $b_{i\gidx}=2C$. 
The number of non-zero interactions $W_{i\gidx,i'\gidx'}$ is estimated as $(\lvert\mathcal{N}\rvert^2 K + K^2 \lvert\mathcal{N}\rvert) / 2$, where the first and second terms are counted in Eq.~(\myref{eq:clustering_HA}) and (\myref{eq:clustering_HB}), respectively. 
An example of the spatial coupling function of the obtained spQUBO is shown in Fig.~\myref{fig:clustering}d. 

\section*{Code availability}

The computer code used for the numerical examples is available at \url{https://github.com/hiroshi-yamashita/spqubo}, which has been archived in Zenodo \mycite{ZenodoSPQUBO}.

\section*{Data availability}

The data generated during the current study are available from, or can be reproduced using, the repository at \url{https://github.com/hiroshi-yamashita/spqubo}, which has been archived in Zenodo \mycite{ZenodoSPQUBO}.

\bibliography{mybib}

\section*{Acknowledgments}
This work is partially supported by JST CREST (JPMJCR18K2), by JST ALCA-Next (JPMJAN23F2), and by JST Moonshot R\&D Program (JPMJMS2021). 

\section*{Author contributions}

H.Y. and H.S. conceived the study. H.Y. developed the formulation, performed the theoretical analysis, and performed the numerical experiments. H.Y. and H.S. discussed the results and wrote the manuscript.

\section*{Competing interests}

The authors declare no competing interests.

\end{multicols}

\newpage

\renewcommand\thesection{\arabic{section}}

\makeatletter
\renewcommand{\@seccntformat}[1]{%
  \ifcsname supp@#1@format\endcsname
    \csname supp@#1@format\endcsname
  \else
    \csname the#1\endcsname\quad
  \fi
}

\newcommand{\supp@section@format}{Supplementary Note~\thesection\quad}
\makeatother

\renewcommand\thesubsection{\arabic{section}.\arabic{subsection}}
\renewcommand{\thefigure}{S\arabic{figure}}
\renewcommand{\thetable}{S\arabic{table}}
\renewcommand{\thetheorem}{S\arabic{theorem}}
\renewcommand{\theequation}{S\arabic{equation}}

\setcounter{section}{0}
\setcounter{figure}{0}
\setcounter{table}{0}
\setcounter{theorem}{0}
\setcounter{equation}{0}

\begin{center}
    {\LARGE \bfseries Supplementary Information for ``Convolutional Formulation of Large-Scale Quadratic Unconstrained Binary Optimization with Dense Interactions'' \par}
\end{center}
\vspace{2em}

\section{Transformation of spQUBO into two-dimensional}
\label{sec-higherorder}

In this section, we prove Theorems \myref{thm:padding-spcm} and \myref{thm:higherorderQUBO} in the text. 
In particular, these can be obtained as a corollary of the following theorem that collapses the configuration domain into two-dimensional preserving the convolutional structure: 

\begin{theorem}

\label{thm:higherorder}
Let $f\colon\Inte^D\to\Real$ be a function that has a locality of $\vran < \vlsh$, and $\mathcal{D}_1, \mathcal{D}_2$ be non-overlapping index sets such that $\mathcal{D}_1\cup \mathcal{D}_2=\{1,\ldots, D\}$ and $\mathcal{D}_1\cap \mathcal{D}_2=\varnothing$. 
Let $\vect{\tilde{\lsh}}\in\Nnum^2$ be the vector whose elements are defined as 
\AL{
\tilde{\lsh}_\gamma=\prod_{\idim\in \mathcal{D}_\gamma}(\lshd+\rand)\LB{eq:lpa}
}
for $\gamma=1,2$.  
Then, there exist a function $f'\colon\Inte^2\to\Real$ and an injective function $\ravx\colon\InteN{\vlsh} \to \InteN{\vect{\tilde{\lsh}}}$ such that $f'$ is periodic with period $\vect{\tilde{\lsh}}$ and 
\AL{
f'(\ravx(\vpos_i)-\ravx(\vpos_j))=f(\vpos_i-\vpos_j) 
}
holds for all $\vpos_i,\vpos_j\in \InteN{\vlsh}$. 

\end{theorem}
\subsection{Mapping coordinate vectors to scalar values}
\label{sec-mapping}

The required map $h'$ transforms $D$-dimensional coordinate vector to a pair of scalar values. 
However, designing this can be reduced to a problem of transformation into a single scalar. 

Let us denote by $\vx_{:K}=(x_1,\ldots,x_K)$ the vector obtained by taking the first $K$ elements of a vector $\vx$. 
We introduce the following notations for $k=1,\ldots,D$: 
\AL{
\vlpa&\equiv\vlsh+\vran\LB{eq:vec-lpa-onedim}\\
\ransp^{K}&\equiv\kdim[-\rand, \rand],\LN{eq:def-ransp}
{\lpa}\lkth&\equiv\kdim\lpad,\LN{eq:def-lpa}
\ravl\lkth(\vx_{:K})&\equiv\ksum x_\idim{\lpa}^{(\idim-1)},\\
\ran\lkth&\equiv \ravl\lkth(\vran_{:K}) = \ksum\rand {\lpa}^{(\idim-1)},
}
where ${\lpa}^{(0)}=\ran^{(0)}=1$ and Eq.~(\myref{eq:def-ransp}) means the Cartesian product. 
The following recurrence relations hold for $\ravl\lkth$ and $\ran\lkth$: 
\AL{
\ravl\lkth(\vx_{:K})&= x_K {\lpa}\mkth+\ravl\mkth(\vx_{:K-1}),\LN{eq:rec-ravl}
\ran\lkth&=\ran_K{\lpa}\mkth+\ran\mkth.\LB{eq:rec-ran}
}
We also use simple notations for $K=D$ as $\ransp\equiv \ransp^D$ and $\lpa\equiv{\lpa}^{(D)}$. 
We also wrap $\ravl\dimth$ to define $\ravm\colon \Inte^D\to \lpa$ as 
\AL{
\ravm(\vect{x})=\mymod{\ravl\dimth(\vect{x})}{\lpa},\LB{eq:onedim-transformed-coordinate}
}
where, for $x\in\Inte$ and $N\in\Nnum$, $\mymod{x}{N}$ denotes the unique $y\in\InteN{N}$ that satisfies $y=x+kN$ for some $k\in\Inte$. 
The wrapped map $\ravm$ preserves the structure of relative positions of local spin pairs well, so that it can be used to collapse the configuration domain. 
Specifically, the following properties hold, whose proofs are given in the following sections: 
\

\begin{prop}

\label{thm:construction}
Let $h$ be the map obtained by restricting the domain of $g$ to $\ransp$, and $I=g(\ransp)=\{g(\vrel)\mid \vrel\in\ransp\}$ be its image.  
Then, the following hold: 

\begin{enumerate}[label=(\Alph*)]
\item $\mymod{\ravm(\vpos_i)-\ravm(\vpos_j)}{\lpa} = \ravm(\vpos_i-\vpos_j)$ for all $\vpos_i ,\vpos_j\in \Inte^D$ 
\item $h$ is injective.
\item $\ravm(\vrel)\not\in\imset$ holds for $\vrel \in \lshsp\setminus \ransp$ where $\lshsp=\ddim (-\lshd,\lshd)$.
\end{enumerate}

\end{prop}
We first present a simpler version of Theorem \myref{thm:higherorder} in one-dimensional case to see how the constructed map $g$ can be used. 

\begin{prop}[Theorem \ref{thm:higherorder} in one-dimensional case]
\label{thm:onedim}
For vectors $\vran < \vlsh\in\Nnum^D$, let $f\colon\Inte^D\to\Real$ be a function that has a locality of $\vran$ and let $\lpa \equiv \ddim(\lshd+\rand)$. 
Then, there exist a function $f': \Inte\to\Real$ with period $\lpa$ and an injective function $\ravx\colon \InteN{\vlsh} \to \InteN{\lpa}$ such that 
\AL{
f'(\ravx(\vpos_i)-\ravx(\vpos_j))=f(\vpos_i-\vpos_j).\LB{eq:onedim-commutativity} 
}
holds for all $\vpos_i,\vpos_j\in \InteN{\vlsh}$. 

\end{prop}

\begin{proof}

Let us define $f'\colon \Inte\to\Real$ as 
\AL{
f'(\rel)= 
\begin{cases}
f(\ravr^{-1}(\mymod{\rel}{\lpa}))& 
(\mymod{\rel}{\lpa} \in \imset)\\
0&
(\mymod{\rel}{\lpa} \not \in \imset)
\end{cases},\LB{eq:onedim-transformed-coupling}
}
where $h$ and $I$ are defined as Proposition \myref{thm:construction}.  
This map is well-defined because $h$ is injective ((B) of Proposition \myref{thm:construction}). 

Let $\ravx$ be the injective map obtained by restricting $\ravm$ to $\InteN{\vlsh}$. 
Then, for any $\vpos_i,\vpos_j\in\InteN{\vlsh}$, it holds 
\AL{
\mymod{\ravx(\vpos_i)-\ravx(\vpos_j)}{\lpa}=g(\vpos_i-\vpos_j)
}
by (A) of Proposition \myref{thm:construction}. 
When $\vpos_i-\vpos_j \in \ransp$, it holds $g(\vpos_i-\vpos_j)\in I$. 
Then, by calculating from the definition of $f'$, Eq.~(\myref{eq:onedim-commutativity}) holds as 
\AL{
&f'(\ravx(\vpos_i)-\ravx(\vpos_j))\NN
&=f(\ravr^{-1}(\mymod{\ravx(\vpos_i)-\ravx(\vpos_j)}{\lpa}))\NN
&=f(\ravr^{-1}(g(\vpos_i-\vpos_j)))\NN
&=f(\vpos_i-\vpos_j).
}
When $\vpos_i-\vpos_j \in \lshsp\setminus\ransp$, it holds $g(\vpos_i-\vpos_j)\not\in\imset$ by (C) of Proposition \myref{thm:construction}. 
Then, we have $f'(\ravx(\vpos_i)-\ravx(\vpos_j))=0$ by definition. 
Then, because $f(\vpos_i-\vpos_j)=0$ holds from the locality of $f$, we obtain Eq.~(\myref{eq:onedim-commutativity}). 

\end{proof}
The proof can be extended to the proof of two-dimensional case as follows: 

\begin{proof}
[Proof of Theorem \ref{thm:higherorder}]
For $D$-dimensional vector $\vect{x}$ and $\gamma\in\{1,2\}$, let us denote by $\vect{x}\grp{\gamma}$ the vector consisting of the elements $x_\idim$ for $\idim\in \mathcal{D}_\gamma$. 
Let $\ravm_\gamma$ and $\imset_\gamma$ be $\ravm$ and $\imset$ of Proposition \myref{thm:construction} when $\vlsh$ and $\vran$ are $\vlsh_\gamma$ and $\vran_\gamma$, respectively. 

Let us combine $\ravm_\gamma$ to define $\ravm(\vpos)=(\ravm_1(\vpos\grp{1}),\ravm_2(\vpos\grp{2}))$. 
In addition, let $h$ be the map obtained by restricting the domain of $g$ to $\ransp$, and $I=g(\ransp)=\{g(\vrel)\mid \vrel\in\ransp\}$ be its image. 
For $\vrel\in\lshsp\setminus\ransp$, we obtain $\ravm_1(\vrel\grp{1})\not\in \imset_1$ or $\ravm_2(\vrel\grp{2})\not\in \imset_2$ from Proposition \myref{thm:construction} and thus $g(\vrel)\not\in\imset$ because $\imset=\imset_1\times\imset_2$. 
Therefore, $\ravm$ and $\ravr$ satisfies three conditions (A), (B), (C) in Proposition \myref{thm:construction} where $\lpa$ is replaced by $\vect{\tilde{\lsh}}$, because (A) and (B) are obvious from the constructions of $g$ and $h$. 

We can complete the proof by following the argument of Proposition \myref{thm:onedim}, using these $g$ and $h$ and substituting $\vect{\tilde{\lsh}}$ for $\lpa$. 

\end{proof}

\subsection{Proof of Proposition \ref{thm:construction}}
\label{sec-proofA}

\begin{lfigure}
\captionsetup[sub]{labelformat=empty}
\begin{center}

\includegraphics[ width=5.7in ]{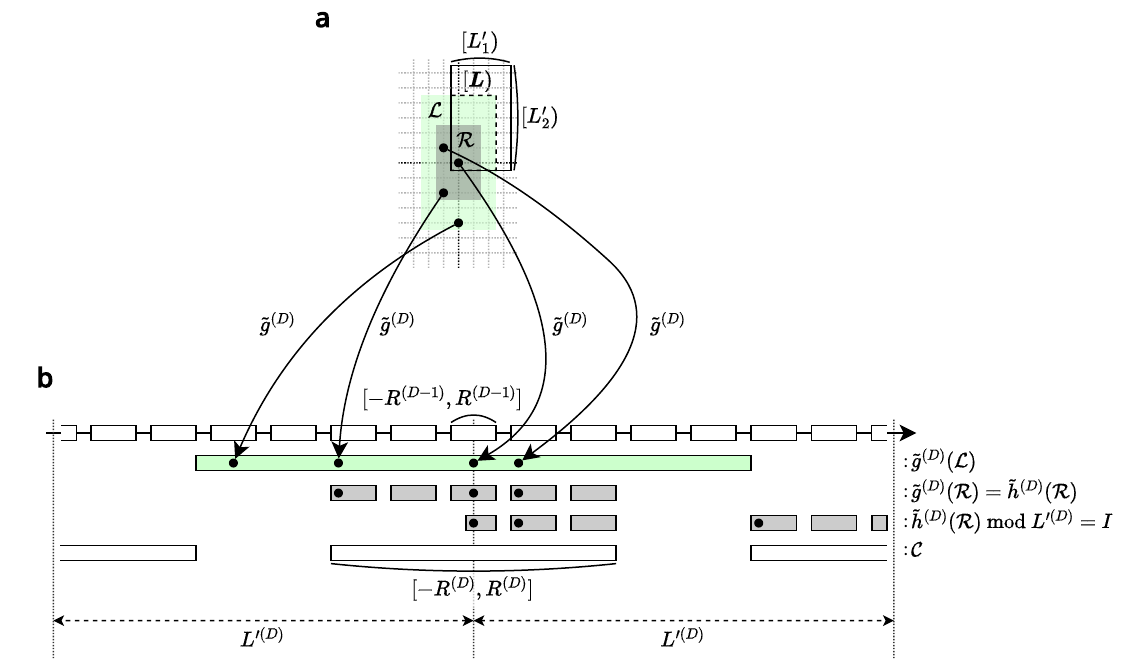}
\end{center}
\caption{
An example of maps used in the transformation of the configuration domain and in the proof of Proposition \myref{thm:construction}.  
(a) a configuration domain for $D=2$. 
The parameters are $\vlsh=(3,5)$, $\vran=(1, 2)$ and $\vlpa=(4,7)$.  
(b) The images of the maps at the last step of induction $K=D=2$. 
The correspondences in each map for $\vx=(0, -4)\in\mathcal{\lsh}$ and $\vx=(-1,-2),(0,0),(-1,1)\in\mathcal{\ran}$ are also shown. 
}
\label{fig-proofs}
\end{lfigure}

Here we prove Proposition \myref{thm:construction}. 
Figure \myref{fig-proofs} shows an example of the maps and their images used in the proof. 

To prove Proposition \myref{thm:construction}, we use the following facts: 

\begin{prop}

\label{prop:propB1}
It holds 
\AL{
2\ran\lkth+1\le {\lpa}\lkth.\LB{eq:propB1}
}

\end{prop}

\begin{proof}

For any $\idim$, it holds 
\AL{
2\rand+1\le \lpad, \LB{eq:propB1-k1}
}
because $\lpad=\lshd+\rand$ and $\rand\le \lshd-1$ by definition. 

The inequality Eq.~(\myref{eq:propB1}) for $K=1$ holds by Eq.~(\myref{eq:propB1-k1}). 
For $K>1$, we can inductively prove it by assuming it for $K-1$ and using the recurrence relation (\myref{eq:rec-ran}) as 
\AL{
2\ran\lkth+1&=2\ran_K {\lpa}\mkth+2\ran\mkth+1\NN
&\le 2\ran_K {\lpa}\mkth+{\lpa}\mkth\NN
&=(2\ran_K+1){\lpa}\mkth\NN
&\le \lpa_K {\lpa}\mkth\NN
&={\lpa}\lkth.\LB{eq:propB1-5}
}

\end{proof}

\begin{prop}

\label{prop:propB2}
Let $\ravlr\lkth$ be the map obtained by restricting the domain of $\ravl\lkth$ to $\ransp^K$. 
Then, $\ravlr\lkth$ is injective and its image is contained in $[-\ran\lkth, \ran\lkth]$. 

\end{prop}

\begin{proof}

We prove it by induction. 
When $K=1$, $\ravlr^{(1)}$ is injective and its image is $\ransp^1=[-\ran^{(1)}, \ran^{(1)}]$, because $\ravlr^{(1)}(x)=x$. 

Let us assume that $K>1$ and that the statement holds for $K-1$. 
We can use the recurrence relation (\myref{eq:rec-ravl}) to calculate $\ravl\lkth$ as 
\AL{
&\ravl\lkth(\vx_{:K})=x_{K} {\lpa}\mkth +\ravl\mkth(\vx_{:K-1}). 
}
When $\vx_{:(K-1)}\in\ransp^{K-1}$, by this recurrence relation and the assumption, it holds that 
\AL{
\ravl\lkth(\vx_{:K})\in [& x_{K} {\lpa}\mkth -\ran\mkth,
\TW{\NN &}
x_{K} {\lpa}\mkth +\ran\mkth]. 
}
By considering the range $x_{K}\in [-\ran_K, \ran_K]$ and using the recurrence relation (\myref{eq:rec-ran}), we can prove that the image of $\ravlr\lkth$ is contained in $[-\ran\lkth, \ran\lkth]$. 
For any integers $n,k>0$, we have 
\AL{
&
((k+n) {\lpa}\mkth-\ran\mkth)-(k {\lpa}\mkth+\ran\mkth)
\NN &
=n{\lpa}\mkth -2\ran\mkth\NN
&\geq{\lpa}\mkth-2\ran\mkth\NN
&>0,
}
where the last inequality is from Proposition \myref{prop:propB1}. 
It tells us that the images of $\ravlr\lkth(x_1,\ldots,x_{K-1},x_{K})$ for different $x_{K}=k$ do not overlap. 
Therefore $\ravlr\lkth$ is injective, due to Eq.~(\myref{eq:rec-ravl}) and the assumption. 

\end{proof}

The proof of each statement in Proposition \myref{thm:construction} is as follows: 

\begin{proof}[Proof of Proposition \ref{thm:construction}]
~\paragraph{Statement (A).}
By definitions, it holds 
\AL{
\TW{&}
\mymod{\ravm(\vpos_i)-\ravm(\vpos_j)}{\lpa}
\TW{\NN}
&= \mymod{\ravl\dimth(\vect{\vpos_i})-\ravl\dimth(\vect{\vpos_j})}{\lpa}
\NN
&=\mymod{\ravl\dimth(\vect{\vpos_i}-\vect{\vpos_j})}{\lpa}\NN
&=\ravm(\vect{\vpos_i}-\vect{\vpos_j}).
}
for $\vpos_i ,\vpos_j\in\Inte^D$. 

~\paragraph{Statement (B).}

Let $\ravlr\dimth$ be the injective function defined as Proposition \myref{prop:propB2}. 
Then, the restricted map $h$ is represented as 
\AL{
\ravr(x)=\mymod{\ravlr\dimth(x)}{{\lpa}\dimth}.
}
This $h$ is also injective because the image of $\ravlr\dimth$ is contained in $[-\ran\dimth, \ran\dimth]$ and, by Proposition \myref{prop:propB1}, we have $\ran\dimth<{\lpa}\dimth-\ran\dimth$. 

~\paragraph{Statement (C).}

Let us assume $\vrel\in\lshsp \setminus \ransp$. 
Then, there exists $K$ such that $\vrel_{:K-1}\in\ransp\mkth$ and $\rel_{K}\not\in [-\ran_K, \ran_K]$. 
For such $K$, it holds $\ravl\mkth(\vrel_{:K-1})\in[-\ran\lkth, \ran\lkth]$ 
by Proposition \myref{prop:propB2}. 
Let us introduce the set $\cfil$ as follows: 
\AL{
\cfil=\{x\in\Inte \mid {}&\exists b \in [-\ran\lkth, \ran\lkth], \TW{\NN
&}\mymod{(x-b)}{{\lpa}\lkth}=0\}.
}

When $\rel_K>\ran_K$, it holds $\ran_K<\rel_K<\lsh_K=\lpa_K-\ran_K$, and thus $\ran_K+1\le \rel_K \le \lpa_K-\ran_K-1$. 
By using Proposition \myref{prop:propB1} and the recurrence relations Eqs.~(\myref{eq:rec-ravl}) and (\myref{eq:rec-ran}) we have 
\AL{
\ravl\lkth(\vrel_{:K})&\in[(\ran_K+1){\lpa}\mkth-\ran\mkth,
\TW{\NN&\quad\;}
(\lpa_K-\ran_K-1){\lpa}\mkth+\ran\mkth]\NN
&\subset(\ran_K{\lpa}\mkth+\ran\mkth,
\TW{\NN&\quad}
(\lpa_K-\ran_K){\lpa}\mkth-\ran\mkth)\NN
&=(\ran\lkth,{\lpa}\lkth-\ran\lkth)\NN
&\subset \Inte \setminus \cfil.
}
Likewise, when $\rel_K<-\ran_K$, it holds $-\lpa_K+\ran_K=-\lsh_K<\rel_K<-\ran_K$, and thus $-\lpa_K+\ran_K+1\le \rel_K\le -\ran_K-1$. 
Thus, we have 
\AL{
\ravl\lkth(\vrel_{:K})&\in[(-\lpa_K+\ran_K+1){\lpa}\mkth-\ran\mkth,
\TW{\NN&\quad\;}
(-\ran_K-1){\lpa}\mkth+\ran\mkth]\NN
&\subset((-\lpa_K+\ran_K){\lpa}\mkth+\ran\mkth,
\TW{\NN&\quad\;}
-\ran_K{\lpa}\mkth-\ran\mkth)\NN
&=(-{\lpa}\lkth+\ran\lkth,-\ran\lkth)\NN
&\subset \Inte \setminus \cfil.
}
In either case, we have $\ravl\lkth(\vrel_{:K})\notin \cfil$. 

Suppose that $\ravm(\vrel)\in\imset$, that is, there exists $\vrel'\in\ransp$ such that $\ravm(\vrel)=\ravm(\vrel')$. 
We have 
\AL{
\ravl\dimth(\tilde{\vrel})&=\Bigg(\sum_{k=K+1}^{D}\tilde{\rel}_k{\lpa}^{(k-1)}\Bigg)+\ravl\lkth(\tilde{\vrel}_{:K})
}
for both $\tilde{\vrel}\in\{\vrel, \vrel'\}$. 
By Proposition \myref{prop:propB2} and $\vrel'_{:K}\in\ransp\lkth$, we have $\ravl\lkth(\vrel'_{:K})\in[-\ran\lkth, \ran\lkth]\subset\cfil$. 
Then, because ${\lpa}^{(k-1)}$ is divisible by ${\lpa}\lkth$ for $k=K+1,\ldots,D$, we have $\ravl\dimth(\vrel)\not \in\cfil$ and $\ravl\dimth(\vrel')\in\cfil$. 

For both $\tilde{\vrel}=\{\vrel, \vrel'\}$, we have 
\AL{
\mymod{\ravm(\tilde{\vrel})}{{\lpa}\lkth} &=\mymod{(\mymod{\ravl\dimth(\tilde{\vrel})}{\lpa})}{{\lpa}\lkth} \NN
&=\mymod{\ravl\dimth(\tilde{\vrel})}{{\lpa}\lkth}
}
because $\lpa$ is divisible by ${\lpa}\lkth$, and thus $\ravl\dimth(\tilde{\vrel})\in\cfil \Leftrightarrow \ravm(\tilde{\vrel})\in\cfil$. 
Since this contradicts the assumption $\ravm(\vrel)=\ravm(\vrel')$, it follows that $\ravm(\vrel)\not\in\imset$. 

\end{proof}

\section{Calculations for spQUBO using DFT}
\label{sec-proof-fft}

Let us denote the two-dimensional DFT of function $\phi\colon\Inte^2\to\Real$ with period $\vlsh$ and its inverse as 
\AL{
\fou\phi(\spos)=\sum_{\vx\in\InteN{\vlsh}}\phi(\vx)W_{\spos,\vx}\\
\ifou\phi(\spos)=\frac{1}{V}\sum_{\vx\in\InteN{\vlsh}}\phi(\vx)W_{\spos,-\vx}
}
where $W_{\spos, \vx}=\exp\AS{-2\pi \iunit \sum_{d=1}^{2}(\spos_d\vx_d/\lsh_d)}$ and $V=\lsh_1\lsh_2$. 

Theorems 3 and 4 can be reduced to the case where $c_i=1$: 

\begin{prop}

\label{thm:fft-mvm-base}
For a spatial coupling function $f$ of an spQUBO, let $a_i$ be elements of the matrix-vector product defined as 
\AL{
a_i = \sumj f(\vdi-\vdj)v_j, 
}
where $\vect{v}\in\Real^N$. 
We define a function $\xi\colon\Inte^2\to\Real$ with period $\vlsh$ as 
\AL{
\xi(\vx) = \begin{cases}
v_i & \text{if } \vx = \vdi, \\ 
0 & \text{otherwise}
\end{cases}
}
for $\vx\in\InteN{\vlsh}$. 
Then, it holds that 
\AL{
a_i = \ifouAS{\fou{\xi}\odot\fou{f}}(\vdi), 
}
where $\odot$ denotes the Hadamard element-wise product. 

\end{prop}

\begin{proof}

Because $\vpos_i$ are all different and by definitions, we can represent $\vect{a}$ as the convolution of $f$ and $\xi$, in particular, for a function 
\AL{
a(\vx_1)=\sum_{\vx_2\in\InteN{\vlsh}}f(\vx_1-\vx_2)\xi(\vx_2), 
}
the $i$-th element of $\vect{a}$ is represented as 
\AL{
a_i&=a(\vpos_i). 
}
Because $f$ has period of $\vlsh$, we can represent the cyclic convolution as 
\AL{
a(\vx)=\ifouAS{{\fou{\xi}} \odot \fou{f}}(\vx), 
}
and this completes the proof. 

\end{proof}
This can be easily extended for the Hamiltonian calculation: 

\begin{proof}
 [Proof of Theorem \myref{thm:fft-hamiltonian}]
Let $\vect{a}$ and $\xi$ be as defined in Proposition \myref{thm:fft-mvm-base}, with $v_i=c_ix_i$, so that $a_i = \ifouAS{{\fou{\xi}}\cdot \fou{f}}(\vpos_i)$. 
Because $\vpos_i$ are all different, we have 
\AL{
\hamil&=\frac{1}{2}\sumi v_i a_i\NN
&=\frac{1}{2}\sum_{\vx\in\InteN{\vlsh}} \xi(\vx)\cdot \AS{\ifouAS{{\fou{\xi}}\cdot \fou{f}}(\vx)}\NN
&=\frac{1}{2V}\sum_{\spos\in\InteN{\vlsh}} \overline{\fou{\xi}(\spos)}\cdot \AS{{\fou{\xi}}(\spos)\cdot \fou{f}(\spos)}\NN
&=\frac{1}{2V}\sum_{\spos\in\InteN{\vlsh}} \|\fou{\xi}(\spos)\|^2\cdot \fou{f}(\spos).
}

\end{proof}
Theorem \myref{thm:fft-mvm} can also be obtained as a corollary of Proposition \myref{thm:fft-mvm-base}. 

\section{MVPs in Ising solvers}
\label{sec-mvms-in-solver}

As mentioned in Introduction, many of the computational principles of physical Ising solvers can be written in dynamical systems \mycite{CAC,CACScaling,MA,SBM}. 
The MVP $J\vect{x}$ argued in the text is incorporated in these systems as below, showing the potential benefits of its improvement. 

If we model the behavior of the coherent Ising machine (CIM) \mycite{CIMPRA2013,CIMScience2016,CIMAPL2020} as classical dynamics with a time interval $\Delta t$, we obtain 
\AL{
\tau\frac{\vect{q}_{k}-\vect{q}_{k-1}}{\Delta t}&=c(-\vect{q}_{k-1}^3+a\vect{q}_{k-1})+J\vect{q}_{k-1},
}
where $\vect{q}_k$ is the system state at the $k$-th iteration, $\tau$ is the time constant, $c$ and $a$ are the other system parameters, and $\vect{q}^n$ means the element-wise $n$-th power of the vector. 
This system can be interpreted as the gradient system or the gradient descent method on the energy function $\hamil(\vect{q})$, defined as 
\AL{
E(\vect{q})&=-\frac{1}{2}\vect{q}\tr J\vect{q}+\vect{\phi}_a(\vect{q}),\\ 
\phi_a(\vect{q})&=\frac{c}{4}\AS{\vect{q}^4-2a\vect{q}^2}.
}
This function simulates the Ising Hamiltonian well when each amplitude is close to one $(q_i)^2\approx 1$. 
It is shown that the performance is improved by adding the controlling terms for this condition \mycite{CAC,CACScaling}. 

In contrast, the simulated bifurcation machine (SBM) \mycite{SBM} simulates the following dynamical system: 
\AL{
\tau_q\frac{\vect{q}_k-\vect{q}_{k-1}}{\Delta t}&=\vect{p}_{k-1}\\
\tau_p\frac{\vect{p}_k-\vect{p}_{k-1}}{\Delta t}&=c(-\vect{q}_{k}^3+a\vect{q}_k) + J \vect{q}_k,
}
where $\vect{q}_k$ and $\vect{p}_k$ are the system states at the $k$-th iteration, and $\tau_p, \tau_q$ are the time constants. 
It is also related to the energy function $E(\vect{q})$, as it simulates the Hamiltonian system defined with $E(\vect{q})$. 

In momentum annealing (MA) \mycite{MA}, two layers $\vect{\sigma}^L,\vect{\sigma}^R$ of spin vectors to be optimized are prepared. 
These two layers of spins are coupled with their coefficient matrix $J$ and additional inter-layer couplings represented by a diagonal coefficient matrix $W$. 
Then, they are updated following the Gibbs sampling procedure, whose update formula is expressed as 
\AL{
\vect{I}_k&=\vect{h}+(J+W)\vect{\sigma}_{k-1},\LN{eq:MA-inte}
\sigma_k&=\sgn \AS{\vect{I}^{(k)}+\frac{T_k}{2}\vect{\Gamma}_k \vect{\sigma}_{k-2}},\LB{eq:MA-flip}
}
where $\vect{\sigma}_k$ represents $\vect{\sigma}^L$ or $\vect{\sigma}^R$ depending on the parity of $k$, $T_k$ is the temperature parameter, and $\Gamma$ is a gamma distribution with the shape and scale parameters set to one. 

\section{Settings in numerical example of placement problem}
\label{sec-placement-numerical}

For the example placement problem used in the text, we consider the region of interest represented as $X=[0, H]\times [0, W]$ for $H,W>0$. 
For $B\in\Nnum$, let us consider the grid points obtained by dividing a square of unit side length into $B^2$ cells. 
Namely, the $i$-th grid point $\vect{p}_i$ corresponds to the integer grid point $\vpos_i = (n, m)\tr$ as 
\AL{
\vect{p}_i=\AS{\frac{n}{B},\frac{m}{B}},
}
where $n=0,\ldots,HB,m=0,\ldots,WB$. 
We assign spins to every grid point, so the number of spins is $O(HWB^2)$. 

We set the parameters for the interaction as follows: 
Let us denote the area of the circle with radius $\rho$ by $S=\pi \rho^2$. 
We assume the point-wise utility $\xi_\vz$ at $\vz$ using the density of the facilities within the area of radius $\rho$ from $\vz$, denoted by $x_\vz=n_{\vz}/S$. 
Specifically, we assume 
\AL{
\xi_\vz=-(\tilde{a}/2)(x_\vz/K-1)^2,
}
where $\tilde{a}$ is the scaling parameter of the utility and $K$ is the reference density. 
The number of facilities in the circle with radius $\rho$ at the reference density is $\nu\equiv SK$. 
Thus, we have 
\AL{
\xi_\vz=-(\tilde{a}/2)(n_{\vz}/\nu-1)^2.
}
For $(a,b)=(\tilde{a}/2\nu^2, \tilde{a}/\nu)$, we can compute the integral of the utility up to the constant similarly to the text as 
\AL{
\Xi &= \int_\vz (-an_\vz^2+bn_\vz) d\vz = -a \sumij W'_{ij}x_ix_j+b\sumi b'_i x_i,
}
where  
\AL{
W'_{ij}&=f(\rel_{ij})=\begin{cases} 
2(\rho^2\theta_{ij} - \rel_{ij}^2\tan\theta_{ij}/4)&(\rel_{ij}<2\rho)\\
0&(\rel_{ij}\ge 2\rho)
\end{cases},\\
b'_i&=S, 
}
for $r_{ij} = \|\vect{p}_i-\vect{p}_j\|$ and $\theta_{ij}=\cos^{-1} (\rel_{ij}/2\rho)$. 
We set $\tilde{a}=\nu/S$ such that the bias values $bb'_i$ become a unit, and set $(H, W, K, \rho)=(5, 5, 2, 0.25)$. 
We used $B=40$ for Fig.~\myref{fig:placement_qubo}, and varied $B=1,2,\ldots, 40$ for Fig.~\myref{fig:fig2}. 
Using Theorem \myref{thm:padding-spcm} in the main text, we converted the problem to the equivalent two-dimensional periodic spQUBO. 
The locality parameter is set to be the minimum number such that it is no less than $2\rho B$. 

We randomly assigned the placement cost for each point as follows: 
Let the placement cost at $\vect{p}_i$ be assigned as $\plcost_i=f(\vect{p}_i)$ for $f\colon X\to\Real$. 
The cost distribution $f$ consists of {\it line} and {\it blob} components. 
In particular, it is defined as 
\AL{
f(x)&\equiv{A}\exp(\tilde{f}(x))-{B},\\ 
\tilde{f}(x)&\equiv\tilde{A}\left(\coline\sum_{i=1}^{\nline}\fline_i(x)\TW{\right.\NN
&\quad\left.}+\sum_{i=1}^{\nblob}\coblob_i\fblob_i(x)\right)-\tilde{B},\\
\fline_i(x)&\equiv\exp\AS{-\frac{\langle x-m_i,v_i\rangle^2}{2(\sline)^2}}\quad(i=1,\ldots,\nline),\\
\fblob_i(x)&\equiv\exp\AS{-\frac{\| x-\mu_i\|^2}{2(\sblob)^2}}\quad(i=1,\ldots,\nblob),
}
with the parameters below: 
There are $\nline=4$ line components determined by the vectors $m_i, v_i \in \Real^2$, which represent the center and the normal direction of the line components and set as 
\AL{
(4m_i)\tr&=(1,2),(2,1),(3,2),(2,3),\\
v_i\tr&=(1,0),(0,1),(1,0),(0,1), 
}
respectively. 
There are also $\nblob$ blob components whose centers are determined by the vectors $\mu_i\in X$, each of which is uniformly randomly sampled from $X$. 
The scales of the line and blob components are determined by the parameters $\sline,\sblob > 0$, respectively. 
$\coblob_i$, which is drawn uniformly randomly from $[-1, +1]$, represents the sign and weight of each blob component, and the balance between lines and blobs is controlled by $\coline$.
$\tilde{f}$ represents the relative cost distribution and we rescaled it to obtain the cost distribution $f$: $\tilde{A},\tilde{B}$ are chosen so that the mean and variance of $\{\tilde{f}(\vpos_i)\}$ are $0$ and $(\sigma^{(f)})^2$, respectively, and $A,B$ are set such that the minimum and the maximum value of $f(p_i)$ are $0$ and $\cmax$, respectively, where $\cmax$ is the scale parameter for costs.
We set the parameters as $(\coline, \sline, \nblob, \sblob, (\sigma^{(f)})^2, \cmax) = (10, 0.05, 1000, 0.3, 1.5, 1.5)$. 

For Figure 4, the parameters for the MA were set as follows: 
\AL{
p_k& = p_0 (1 - k / T)\\ 
c_k& = 1\\ 
T_k& =  \frac{T_0\theta_1}{\theta_1 + \log(1 + \theta_2(k-1)/T)}, 
}
where $p_k, c_k$ and $T_k$ are the same parameter as the original study \mycite{MA} of MA for each $k$-th step, $T=10000$ is the number of steps, and $(T_0, p_0, \theta_1, \theta_2) = (0.1, 0.05, 0.001, 0.1)$. 
We set the inter-layer coupling strengths as $w_i = \sum_j |J_{ij}|/2$ for all spins, which are a more simplified choice than the original study. 

\section{Settings in numerical example of clustering problem}
\label{sec-clustering-synthetic}

For the example clustering problem used in the text, we used the following settings. 
For $B\in\Nnum$, let us consider the grid points obtained by equally dividing the unit square $[0, 1]^2$ into $(B-1)^2$ cells. 
Namely, each point $x_i$ is represented as 
\AL{
\vx_i=\AS{\frac{n}{B-1},\frac{m}{B-1}}
}
for $n,m\in\{0,\ldots,B-1\}$. 
We consider $K$ clusters, and generated $M$ points for each the $k$-th cluster $C_k$ as follows: 
Let us suppose the center $\vect{v}_k\in\Real^2$ and the scale parameter $\sigma_k>0$ for each $k$. 
We compute the density parameter 
\AL{
f_{k,i}=\frac{1}{\sqrt{2\pi\sigma_k^2}}\exp\AS{-\frac{(x-v_k)}{2\sigma_k^2}} 
}
for each $k$ and grid point $x_i$. 
The cluster $K_i$ is assigned to each grid point $x_i$ such that it maximizes the density parameter: 
\AL{
K_i=\argmax_k f_{k,i}. 
}
For each $k$-th cluster, we sample $M$ points without replacement using selection probabilities proportional to $f_{k,i}$. 

Specifically, we set $(B,K,M)=(51,7,100)$. 
In addition, we used the parameters for the density parameter given by 
\AL{
v_k&=0.1v'_k\\ 
({v'_k})\tr&= (1,1),(1,9),(9,1),(9,9),\TW{\NN
&\quad\quad} (3,5),(6,3),(6,7)
}
and $\sigma_k=0.1$ for all $k$. 
The weight parameter $C$ for the Hamiltonian $H=H_{\mathrm{A}}+CH_{\mathrm{B}}$ is set to $C=50$. 
The parameters for the MA were set as follows: 
\AL{
p_k& = \max (0, p_0 - k / T)\\ 
c_k& = \min (1, \sqrt{\theta_1 k / T})\\ 
T_k& =  \frac{T_0}{ \log(1 + k)}, 
}
where $p_k, c_k$ and $T_k$ are the same parameter as the original study \mycite{MA} of MA for each $k$-th step, $T=1000$ is the number of steps, and $(T_0, p_0, \theta_1) = (1, 0.5, 2)$. 
The inter-layer coupling strengths are set similarly as the case of the placement problem. 

\end{document}